\documentclass[journal,oneside,web]{ieeecolor}
\usepackage{generic}
\usepackage{cite}
\usepackage{amsmath,amssymb,amsfonts}
\usepackage{algorithmic}
\usepackage{graphicx}
\usepackage{textcomp}

\usepackage{balance}

\newtheorem{thm}{Theorem}
\newtheorem{lem}{Lemma}

\newtheorem{prob}{Problem}
\newtheorem{assum}{Assumption}
\newtheorem{defn}{Definition}

\newtheorem{rem}{Remark}
\usepackage{balance}
\usepackage{bm}
\usepackage{float}

\usepackage{hyperref}

\def\BibTeX{{\rm B\kern-.05em{\sc i\kern-.025em b}\kern-.08em
    T\kern-.1667em\lower.7ex\hbox{E}\kern-.125emX}}
\begin{document}

\title{$H_2$ suboptimal containment control of homogeneous and heterogeneous multi-agent systems
    }
\author{Yuan Gao, \IEEEmembership{Graduate Student Member, IEEE}, Junjie Jiao, \IEEEmembership{Member, IEEE},\\ Zhongkui Li, \IEEEmembership{Senior Member, IEEE}, and Sandra Hirche, \IEEEmembership{Fellow, IEEE} \\
\thanks{Y. Gao, J. Jiao, and S. Hirche are with the Chair of Information-oriented Control, TUM School of Computation, Information, and Technology, Technical University of Munich, 81669, Munich, Germany
        (Email: {\tt\small  ge54sem@tum.de; j.jiao@tum.de;  hirche@tum.de}).
}
\thanks{Z. Li is with the College of Engineering, Peking University, Beijing, 100871, China (Email:  {\tt\small  zhongkli@pku.edu.cn}).
}
    }

\maketitle

\begin{abstract}
This paper deals with the  $H_2$ suboptimal state containment control problem for homogeneous linear multi-agent systems and the $H_2$ suboptimal output containment control problem for heterogeneous linear multi-agent systems.
For both problems, given multiple autonomous leaders and a number of followers, we introduce suitable performance outputs and an associated $H_2$ cost functional, respectively. The aim is to design a distributed protocol by dynamic output feedback that achieves state/output containment control while the associated $H_2$ cost is smaller than an a priori given upper bound. 
To this end, we first show that  the $H_2$ suboptimal state/output containment control problem can be equivalently transformed into $H_2$ suboptimal control problems for a set of independent systems. Based on this, design methods are then provided to compute such distributed dynamic output feedback protocols. Simulation examples are provided to illustrate the performance of our proposed protocols.
\end{abstract}

\begin{IEEEkeywords}
Containment control, distributed control, $H_2$ optimal control, multi-agent systems, suboptimality.
\end{IEEEkeywords}

\section{Introduction}
\label{sec:introduction}
\IEEEPARstart{I}{n} the past two decades, a significant amount of attention has been given to distributed control of multi-agent systems due to its wide range of applications, e.g., formation control~\cite{fax2004information}, flocking \cite{olfati2006flocking}, smart grids \cite{dorfler2013synchronization}, and intelligent transport systems \cite{besselink2017string}.
One of the fundamental research problems in this framework is consensus \cite{olfati2004consensus}. Based on the number of leaders in the network, consensus problems can be categorized into leaderless consensus \cite{li2009consensus}, leader-follower consensus (one leader) \cite{ni2010leader}, and containment control (multiple leaders) \cite{li2013distributed}. 

In scenarios with multiple leaders, the problem of containment control arises, where the objective is to drive the states/outputs of the followers into a convex hull formed by the states/outputs of the leaders.  Containment control, which is inspired by natural phenomena, constitutes a crucial technique applicable to a diverse range of practical applications, including cooperative migration \cite{dong2015formation} and warehouse management \cite{zhou2022intelligent}. Existing literature has studied {\em state} containment control of {\em homogeneous} systems for single-integrators under fixed and switching network topologies with distributed static state protocol \cite{cao2009containment},  for double-integrators consisting of stationary and dynamic leaders with static control algorithms \cite{cao2010distributed}, and for general linear continuous and discrete systems under directed fixed topology by distributed protocols using static state feedback and dynamic output feedback  \cite{li2017cooperative}. %
However, the above works did not take into consideration performance criteria. 

In practice, agent dynamics are subjected to external disturbances, which can significantly deteriorate the performance of multi-agent systems. To address this issue, research has been focused on seeking performance requirements for {\em state} containment control of {\em homogeneous} multi-agent systems. For second-order multi-agent systems, a distributed static state control protocol is proposed in \cite{wang2015robust} to address a robust state containment control problem with a prescribed $H_\infty$ performance. For general linear multi-agent systems, a robust $H_\infty$ state containment control problem over Markovian switching topologies is solved in~\cite{wang2018h} using static state feedback. 
An observer-based state containment control protocol is proposed in~\cite{wang2020distributed} with communication time delay over switching topologies by guaranteeing certain $H_\infty$ performance. Compared to continuous-time systems, the finite-horizon $H_\infty$ containment control is investigated in~\cite{chen2018mathcal} for general discrete-time multi-agent systems with multiple leaders using an event-based distributed controller and a state observer. 
 The above-mentioned works only focus on the $H_\infty$ performance, which indicates the system's robustness to the worst-case scenario in the presence of external disturbances.

Meanwhile, in the context of {\em heterogeneous} multi-agent systems, the models of agents can vary due to heterogeneity in system dynamics and/or state dimensions. For heterogeneous multi-agent systems in which the followers have the same state dimensions, {\em state} containment control arises. For example, in~\cite{zheng2014containment}, state containment control is considered for a mixture of leaders with single-integrators and followers with double-integrators. More recently, using a cooperative output regulation framework, the state containment control problem by state feedback is studied in \cite{haghshenas2015containment} for linear high-order heterogeneous multi-agent systems with nonidentical followers and identical leaders. 
For general heterogeneous agents, however, even the state dimensions of the agents can vary. In this regard,  state containment control is not reasonable anymore, and {\em output} containment control should be considered. To address this problem,  distributed protocols using static state feedback and static output feedback are proposed in \cite{zuo2017output}, based on internal model principles, to achieve output containment control. Compared to \cite{zuo2017output},  a dynamic-output-based distributed protocol is designed in \cite{qin2018output}  to achieve output containment control with fixed and switching networks. 
Note that the proposed control protocols above did not address performance requirements in the presence of external disturbances.  

To address the above issue, research has been focused on seeking performance requirements for containment control of {\em heterogeneous} multi-agent systems. A distributed static output protocol is used in \cite{atrianfar2022robust} to address $H_\infty$   {\em state}  containment control problem with structured uncertainty and external disturbances. However, in \cite{atrianfar2022robust}, the state dimensions of the agent are required to be the same. In \cite{yuan2018h}, distributed control protocols via dynamic output feedback are established to achieve {\em output} containment control with optimized $H_\infty$ disturbance attenuation performance.   It's worth noting that these studies only consider the $H_\infty$ performance index for the worst-case scenario.

So far, in the literature,  little attention has been paid to designing distributed protocols for achieving containment control of multi-agent systems while considering the $H_2$ performance. It is worth mentioning that considerable effort has been invested in guaranteeing $H_2$ performance for leaderless and leader-follower consensus. The consensus control problem is explored for general homogeneous linear multi-agent systems with undirected graphs \cite{li2011h} and directed graphs  \cite{wang2014distributed}, focusing on $H_2$ performance regions.
Considering the minimization of a given $H_2$ cost criterion instead of the $H_2$ performance region,  suboptimal distributed protocols by static state feedback in \cite{jiao2018suboptimality} and by dynamic output feedback in \cite{jiao2020suboptimality, yuan2023I} are established for homogeneous multi-agent systems. Later on,  the results of \cite{jiao2020suboptimality} are extended to the case of heterogeneous multi-agent systems and an $H_2$ suboptimal dynamic protocol is proposed in \cite{jiao2021h2} to achieve output consensus. Overall, the above works regarding $H_2$ performance did not consider the case of multi-agent systems with multiple leaders, i.e., the $H_2$ containment control problem.

\subsection{Contribution and Structure}
Motivated by the above, the present paper deals with the problem of $H_2$ optimal {\em state} containment control problem for {\em homogeneous} linear multi-agent systems and the problem of $H_2$ optimal {\em output} containment control  for {\em heterogeneous} linear multi-agent systems. 
The objective is to design distributed protocols by dynamic 
 output feedback that achieves state/output containment control while minimizing an associated $H_2$ cost functional. Due to the communication constraints among the agents, however, these problems are non-convex in general. A closed-form solution has not yet been given in the literature and may not even exist. Hence, the present paper then addresses an alternative form of these two problems that  involves   {\em suboptimality}. 
More concretely, we aim at finding  distributed protocols by dynamic output feedback that achieve state/output containment control, respectively, while guaranteeing the associated $H_2$ cost is smaller than an a priori given upper bound. The main contributions of this paper are as follows.
\begin{enumerate}
    \item For the homogeneous multi-agent system case, we present a novel distributed dynamic protocol for  achieving $H_2$ suboptimal state containment control by using output feedback. This generalizes our previous results using static state feedback in \cite{yuan2023} by modifying the protocol from \cite{trentelman2013robust} and extends the separation principle results in \cite{jiao2020suboptimality}.
    \item We then extend the results on $H_2$ suboptimal state containment control by dynamic output feedback for homogeneous systems to those on $H_2$ suboptimal output containment control for heterogeneous systems. 
	\item For both problems, by introducing suitable performance outputs, we show that the $H_2$ suboptimal state/output containment control problem can be equivalently recast as   $H_2$ suboptimal control problems of a set of independent systems.
	\item  We then propose design methods for computing $H_2$ suboptimal distributed dynamic output feedback protocols for homogeneous and heterogeneous multi-agent systems, respectively.
\end{enumerate}

The rest of this paper is organized as follows. 
In Section~\ref{sec_preliminaries}, basic notations and graph theories are reviewed, and an ${H}_2$ suboptimal control problem by dynamic output feedback is studied. 
In Section \ref{sec_homogeneous}, we first formulate the distributed $H_2$ suboptimal state containment control problem by dynamic output feedback and  then provide a design method for computing one such protocol. 
After that, in Section \ref{sec_heterogenous}, the distributed $H_2$ suboptimal output containment control problem is first  formulated and  a design procedure for computing a suboptimal dynamic protocol is then proposed. Section \ref{sec_examples} contains two simulation examples to illustrate the performance of our proposed protocols for both cases. Finally, Section \ref{sec_conclusion} concludes this paper.

\section{ Preliminaries}\label{sec_preliminaries}
\subsection{Notation and Graph Theory}\label{subsec_notation}
In this paper, $\mathbb{R}$ represents the field of real numbers, $\mathbb{R}^n$ represents the space of $n$ dimensional real vectors, and $\mathbb{R}^{m \times n}$ represents the space of $m\times n$ real matrices. The identity matrix of size $n\times n$ is represented by $I_n$. The superscript $\top$ means the transpose of a real vector or matrix. We use ${\rm tr} (A)$ to denote the trace of the square matrix $A$. A matrix is called Hurwitz (or stable) if all its eigenvalues have negative real parts. For a symmetric matrix $P$, we denoted $P > 0$ if $P$ is positive definite and $P < 0$ if $P$ is negative definite. We use ${\rm diag}(d_1, \dots, d_n)$ to denoted the $n\times n$ diagonal matrix with $d_1, \dots, d_n$ on the diagonal. For matrices $M_1,\dots, M_m$, the block diagonal matrix with diagonal blocks $M_i$ is denoted by $ {\rm {\rm blockdiag}}(M_1,\dots, M_m)$. The Kronecker product of matrices $A$ and $B$ is represented by $A\otimes B$. 
For a set $X = \left\{x_1, ... , x_n\right\}$ in $V \subseteq \mathbb{R}^p$, its convex hull ${\rm co}(X)$ is defined as 
\begin{equation*}
    {\rm co}(X) = \left\{\sum_{i=1}^n  \alpha_i x_i\ |\  x_i \in V, \alpha_i \geq 0, \sum_{i=1}^n \alpha_i = 1\right\}.
\end{equation*} 

In graph theory, a directed graph is denoted by $\mathcal{G} = (\mathcal{V}, \mathcal{E})$, where $\mathcal{V} = \{ 1,\ldots, N \}$ is the node set  and $\mathcal{E} = \{ {e_1},\ldots, {e_M} \}$ is the edge set satisfying $\mathcal{E} \subset \mathcal{V} \times \mathcal{V}$.  An edge from node $i$ to node $j$ is represented by the pair $(i,j) \in \mathcal{E}$. A graph is said to be undirected if $(i, j) \in \mathcal{E}$ implies $(j, i) \in \mathcal{E}$. A graph is simple if $(i, i) \notin \mathcal{E}$ which means no self-loops. The adjacency matrix $\mathcal{A} = [a_{ij}] \in \mathbb{R}^{N\times N}$ with non-negative elements $a_{ij}$ is defined as follows: $a_{ii} = 0$, $a_{ij} = 1$ if $(j, i) \in \mathcal{E}$, and $a_{ij} = 0$ otherwise. Subsequently, the Laplacian matrix $L = [L_{ij}] \in \mathbb{R}^{N\times N}$ of $\mathcal{G}$ is defined as $L_{ii} = \sum_{j=1}^{N} a_{ij}$ and $L_{ij} = - a_{ij}$. 
Note that $L = \mathcal{D} - \mathcal{A}$, where $\mathcal{D} = {\rm diag}(d_1,\dots, d_N)$ is the degree matrix of $\mathcal{G}$ with $d_{i} = \sum_{j=1}^{N} a_{ij}$. 

\subsection{\texorpdfstring{${H_2}$}~ Suboptimal Control  by Dynamic Output Feedback for Linear Systems}\label{subsec_single_sys}

This subsection considers the  ${H}_2$ suboptimal control by dynamic output feedback for a single linear system. In particular, we generalize the results in  \cite[Lemma 2]{jiao2020suboptimalityarXiv}.

Consider the linear system
\begin{equation}\label{sys_xyz}
	\begin{aligned} 
		\dot{x} & = \bar{A} x + \bar{B} u + \bar{E} d,\\
		y &= \bar{C}_1 x + \bar{D}_1 d, \\
		z &= \bar{C}_2 x + \bar{D}_2 u,
	\end{aligned}
\end{equation}
where $x \in \mathbb{R}^n$, $u \in \mathbb{R}^m$, $d \in \mathbb{R}^q$, $y \in \mathbb{R}^r$, and $z \in \mathbb{R}^p$ are, respectively, the state, the coupling input, the unknown external disturbance, the measured output, and the output to be controlled. The matrices $\bar{A}$, $\bar{B}$, $\bar{C_1}$, $\bar{C_2}$, $\bar{D_1}$, $\bar{D_2}$, and $\bar{E}$ are of suitable dimensions. Throughout this subsection, we assume that the pair $(\bar{A}, \bar{B})$ is stabilizable and that the pair $(\bar{C}_1, \bar{A})$ is detectable.
Next, we consider the case that the system \eqref{sys_xyz} is controlled by a dynamic output feedback controller 
\begin{equation}\label{dyna_w}
	\begin{aligned}
		\dot{w} &= \bar{A} w + \bar{B} u + G \left(\bar{C}_1 w - y\right), \\
		u &= F w,
	\end{aligned}
\end{equation}
where  $w \in \mathbb{R}^n$ is the state of the controller and $F \in \mathbb{R}^{m \times n}$, $G \in \mathbb{R}^{n \times r}$ are feedback gain matrices to be designed. By interconnecting the system \eqref{sys_xyz} and the controller \eqref{dyna_w}, the controlled system is as follows
\begin{equation}\label{sys_dyna_w}
	\begin{aligned}
		\begin{bmatrix}
			\dot{x}\\
			\dot{w}
		\end{bmatrix} 
		& =
		\begin{bmatrix}
			\bar{A} & \bar{B}F \\
			G \bar{C}_1 & \bar{A} +  \bar{B}F + G \bar{C}_1  
		\end{bmatrix} 
		\begin{bmatrix}
			x \\
			w
		\end{bmatrix}
		+
		\begin{bmatrix}
			\bar{E} \\
			-G \bar{D}_1	
		\end{bmatrix}d ,
		\\
		z & =
		\begin{bmatrix}
			\bar{C}_2 & \bar{D}_2 F 
		\end{bmatrix}
		\begin{bmatrix}
			x \\
			w
		\end{bmatrix}.
	\end{aligned} 
\end{equation}
Denote ${A}_a  =  	
	\begin{bmatrix}
		\bar{A} & \bar{B}F \\
		G \bar{C}_1 & \bar{A} + \bar{B}F +G \bar{C}_1 
	\end{bmatrix}$, ${E}_a  =
	\begin{bmatrix}
		\bar{E} \\
		-G \bar{D}_1	
	\end{bmatrix}$, and
	${C}_a = 
	\begin{bmatrix}
		\bar{C}_2 & \bar{D}_2 F 
	\end{bmatrix}.$ 
Then the impulse response matrix from disturbance $d$ to output $z$ is equal to ${T}_{F, G}(t) = {C}_a e^{{A}_a t}{E}_a $.
Subsequently, the associated ${H}_2$ cost functional is given by
	$J(F, G)  := \int_{0}^{\infty} 
	\text{tr} \left[ {T}_{F, G}^{\top}(t) {T}_{F, G}(t) \right] dt.$

The ${H}_2$ suboptimal control problem by dynamic output feedback for the linear system \eqref{sys_xyz} is the problem of finding a controller \eqref{dyna_w} which internally stabilizes the controlled system \eqref{sys_dyna_w} while the associated cost $J(F, G)$ is smaller than an a priori given upper bound.

The following lemma provides conditions under which a controller \eqref{dyna_w} is $H_2$ suboptimal for the system \eqref{sys_xyz}.
\begin{lem}\label{lem_1}
	Let $\gamma > 0$ be a given tolerance.	Assume that $\bar{D}_1\bar{E}^{\top}  =0$, $\bar{D}_2^{\top} \bar{C}_2 =0$, $\bar{D}_1  \bar{D}_1^{\top} >0$, and $\bar{D}_2^\top \bar{D}_2 >0$. 
	Let  $F\in \mathbb{R}^{m \times n}$. Suppose that there exists $P > 0$ and $Q > 0$ satisfying
	\begin{equation*}\label{are_P}
 \begin{aligned}
&	(\bar{A} + \bar{B} F)^\top P + P (\bar{A} + \bar{B} F) \\
&\qquad\qquad\qquad\qquad + (\bar{C}_2 + \bar{D}_2 F)^\top  (\bar{C}_2 + \bar{D}_2 F)   < 0,  
 \end{aligned} 
	\end{equation*}
	\begin{equation*}\label{are_Q}
	\bar{A} Q + Q \bar{A}^\top - Q \bar{C}_1^\top(\bar{D}_1\bar{D}_1^{\top})^{-1}  \bar{C}_1 Q + \bar{E} \bar{E}^\top  <0.
	\end{equation*}
	If, moreover, the inequality 
	\begin{equation*}\label{gamma_P_Q}	
		{\rm tr} \left(\bar{C}_1Q   P Q\bar{C}_1^\top (\bar{D}_1\bar{D}_1^{\top})^{-1} \right) + {\rm tr} \left( \bar{C}_2 Q  \bar{C}_2^\top  \right) < \gamma
	\end{equation*}
	holds, then the controller \eqref{dyna_w} with  $F$ and  $G =- Q\bar{C}_1^\top(\bar{D}_1\bar{D}_1^{\top})^{-1}$  internally stabilizes the system~\eqref{sys_dyna_w} and  $J(F, G) <\gamma$.
\end{lem}

The detailed proof of Lemma~\ref{lem_1}  follows from that of~\cite[Lemma 2]{jiao2020suboptimalityarXiv} and is omitted here.
Indeed, Lemma \ref{lem_1} is a generalization of~\cite[Lemma 2]{jiao2020suboptimalityarXiv} with a relaxed assumption $\bar{D}_1  \bar{D}_1^{\top} >0$. 

\section{\texorpdfstring{${H_2}$}~  suboptimal state containment control of homogeneous linear multi-agent systems}\label{sec_homogeneous}
In this section, we consider the $H_2$ suboptimal state
containment control problem by dynamic output feedback for homogeneous multi-agent systems. 
We  first formulate the problem in Subsection \ref{subsec_problem_homo}. Then in Subsection \ref{subsec_design_homo}, we solve this problem by designing a  distributed $H_2$ suboptimal control protocol.

\subsection{Problem Formulation}\label{subsec_problem_homo}
Consider a multi-agent system with  $N$ agents consisting of $M$ followers subjected to external disturbances and $N-M$ autonomous leaders. 
For clarity and ease of reference, without loss of generality, we assign the labels $1$ through $M$ to the followers and the labels $M+1$ through $N$ to the leaders, respectively. 
We denote the follower set to be $\mathcal{F} \overset{\Delta}{=} \left\{1, \ldots, M\right\}$, while the   leader set is denoted by  $\mathcal{L} \overset{\Delta}{=} \left\{M + 1, ..., N\right\}$.  

The dynamics of the $i$th leader is represented by the linear time-invariant  system 
\begin{equation}
	\begin{aligned} \label{containment_leadersdynamic}
		\dot{x}_i\left(t\right) &= A  x_i\left(t\right),\\
		y_i\left(t\right) &= C_1 x_i\left(t\right), \\
		z_i\left(t\right) &= C_2 x_i\left(t\right),
	\end{aligned}
	\qquad i \in  \mathcal{L},
\end{equation} 
and the dynamics of the $i$th followers is denoted by
\begin{equation}\label{containment_followersdynamic}
	\begin{aligned} 
		\dot{x}_i\left(t\right) & = A x_i\left(t\right)  + B u_i\left(t\right)  + E d_i\left(t\right),\\
		y_i\left(t\right) & = C_1 x_i\left(t\right)  + D_1 d_i\left(t\right),\\
		z_i\left(t\right) & = C_2 x_i\left(t\right)  + D_2 u_i\left(t\right),
	\end{aligned}\qquad i \in \mathcal{F},
\end{equation}
where $x_i \in \mathbb{R}^n$, $u_i \in \mathbb{R}^m$, $d_i \in \mathbb{R}^q$, $y_i \in \mathbb{R}^r$ and $z_i \in \mathbb{R}^p$ are, respectively, the state, the coupling input, the unknown external disturbance, the measured output, and the output to be controlled. The matrices $A$, $B$, $C_1$, $C_2$, $D_1$, $D_2$, and $E$ are of compatible dimensions. Throughout this section, it is assumed that the pair $\left(A, B\right)$ is stabilizable and the pair $\left(C_1, A\right)$ is detectable. 

In this section, we also assume that each follower has access to the relative output measurements of its neighbors and consider the state containment problem using dynamic output feedback. Thus, we consider the case that the leaders \eqref{containment_leadersdynamic} and followers \eqref{containment_followersdynamic} are interconnected by a distributed observed-based  dynamic output feedback protocol of the form 
\begin{equation}\label{containment_dynamicprotocol}
\begin{aligned}
	\dot{w}_i & =  A w_i + B F\left( \sum_{j=1}^M  a_{ij}(w_i-w_j) + \sum_{j={M+1}}^N  a_{ij} w_i \right)  \\
  & 
  \qquad + G\left( C_1w_i-  \sum_{j=1}^N a_{ij} (y_i -y_j) \right),\\
		u_i & = F w_i, \qquad i \in \mathcal{F},
\end{aligned}
\end{equation} 
where $G \in \mathbb{R}^{n \times r}$ and $F \in \mathbb{R}^{m \times n}$ are local feedback gains  to be designed, $w_i$ is the state of the protocol,
and $a_{ij}$ represents the $ij$-th entry of the adjacency matrix $\mathcal{A}$ associated with graph $\mathcal{G}$ which satisfies the following assumption.
\begin{assum}\label{Assumation 1}
The leaders receive no information from any followers. However, each of the leaders shares its  information to at least one of the followers. The communication graph between the $M$ followers is connected, simple, and undirected.
\end{assum}

Accordingly, the Laplacian matrix associated with graph $\mathcal{G}$  can   be partitioned as
\begin{equation}\label{laplacian}
	L = \left[
	\begin{array}{cc}
		L_1 & L_2\\
		0_{(N-M)\times M} & 0_{(N-M)\times (N-M)}
	\end{array}
	\right],
\end{equation}
where ${L}_1 \in \mathbb{R}^{M\times M}$ and ${L}_2 \in \mathbb{R}^{M\times\left(N-M\right)}$.

Denote $\bm{x_f} = \left[x_1^{\top}, \ldots, x_M^{\top}\right]^{\top}$, $\bm{x_l} =\left[x_{M+1}^{\top}, \ldots, x_N^{\top}\right]^{\top}$, $\bm{y_f}$\\
$=\left[y_1^{\top}, \ldots, y_M^{\top}\right]^{\top}$,  $\bm{y_l} =\left[y_{M+1}^{\top}, \ldots, y_N^{\top}\right]^{\top}$, $\bm{u} =
[u_1^{\top},\ldots,$\\$ u_M^{\top}]^{\top}$, $\bm{d} =\left[d_1^{\top}, \ldots, d_M^{\top}\right]^{\top}$, $\bm{w_f} = \left[w_1^{\top}, \ldots, w_M^{\top}\right]^{\top}$, $\bm{z_f} = \left[z_1^{\top}, \ldots, z_M^{\top}\right]^{\top}$, and $\bm{z_l} =[z_{M+1}^{\top},\ldots,z_N^{\top}]^{\top}$. 
We can then write the agents dynamics \eqref{containment_leadersdynamic} and \eqref{containment_followersdynamic} in compact form as
\begin{equation}\label{compact_systems_dynamic}
	\begin{aligned} 
		\bm{\dot{x}_l} &= (I_{N-M} \otimes A) \bm{x_l},\\
		\bm{y_l} &= (I_{N-M} \otimes C_1)\bm{x_l}, \\
		\bm{z_l} &= (I_{N-M} \otimes C_2)\bm{x_l}, \\
		\bm{\dot{x}_f} &= (I_{M} \otimes A) \bm{x_f} + (I_M \otimes B) \bm{u} + (I_M \otimes E) \bm{d},\\
		\bm{y_f} &= (I_M \otimes C_1)\bm{x_f} + (I_M \otimes D_1) \bm{d}\\
		\bm{z_f} &= (I_M \otimes C_2)\bm{x_f} + (I_M \otimes D_2) \bm{u}.
	\end{aligned}
\end{equation} 
Correspondingly, the protocol \eqref{containment_dynamicprotocol} can be written as
\begin{equation}\label{containment_dynamicprotocol_compact}
	\begin{aligned}
		\bm{\dot{w}_f} &   = ( I_{M} \otimes (A + GC_1)) \bm{w_f} + ({L_1} \otimes BF) \bm{w_f}   \\ & \qquad - ({L_1} \otimes G )\bm{y_f} - ({L_2} \otimes G )\bm{y_l} ,\\
		\bm{u} & =(I_{M} \otimes F)\bm{w_f}.
	\end{aligned}
\end{equation}
First of all, we want the dynamic protocol \eqref{containment_dynamicprotocol} to achieve   state containment control   for agents \eqref{containment_leadersdynamic} and \eqref{containment_followersdynamic} without taking into account external disturbances. Here, state containment control means that the states of the followers   converge to the convex hull $\boldsymbol{\omega_x}$ of the states of the leaders~\cite{li2013distributed}, %
which is defined as 
	\begin{equation}\label{convex_hull}
		\boldsymbol{\omega_x}(t) \overset{\Delta}{=} \left(-{L}_1^{-1} {L}_2 \otimes e^{At}\right)
		\left[
		\begin{array}{c}
			x_{M+1}\left(0\right)\\
			\vdots\\
			x_N\left(0\right)
		\end{array}
		\right],
	\end{equation}
where $x_{M+1}\left(0\right), \ldots,  x_N\left(0\right)$ are the initial states of the leaders and the sum of each row of $-{L}_1^{-1} {L}_2$ equal to $1$ which is derived from the following lemma.
\begin{lem}[\cite{meng2010distributed}]\label{lemma8}
	Under Assumption \ref{Assumation 1},  ${L}_1$ is positive definite and each row of $-{L}_1^{-1} {L}_2$ has its sum equal to $1$.
\end{lem}

\begin{defn}\label{definition2}
	The protocol \eqref{containment_dynamicprotocol_compact} is said to achieve   state containment control for   multi-agent system \eqref{compact_systems_dynamic} if the states of the followers $\bm{x_f}$ converge into the convex hull formed by the states of the leaders $\bm{x_l}$ and  the protocol state   $\bm{w_f}$  goes to zero, i.e.,  $\bm{x_f}(t) \rightarrow \boldsymbol{\omega_x}(t)$ and $\bm{w_f}(t) \rightarrow 0$ as $t \rightarrow \infty$. 
\end{defn}

To proceed, we introduce a new error state variable for each follower as $\xi_{xi} = \sum_{j=1}^N a_{ij} (x_i - x_j), i \in \mathcal{F},$
and a new error state variable for each protocol of followers as $\xi_{wi} = \sum_{j=1}^M a_{ij} (w_i - w_j)+\sum_{j={M+1}}^N  a_{ij} w_i, i \in \mathcal{F}.$
Denote $\bm{\xi_x} =\left[\xi_{x1}^{\top}, \ldots, \xi_{xM}^{\top}\right]^{\top}$ and $\bm{\xi_w} =\left[\xi_{w1}^{\top}, \ldots, \xi_{wM}^{\top}\right]^{\top}$,
 we then have
\begin{equation*}\label{xi_dynamic}
    \begin{aligned}
	\bm{\xi_x} &= (L_1 \otimes I_n) \bm{x_f} + (L_2 \otimes I_n) \bm{x_l}, \\
	\bm{\xi_w} &= (L_1 \otimes I_n) \bm{w_f}.
    \end{aligned}
\end{equation*}
Note that, whenever $\bm{\xi_x}$  and $ \bm{\xi_w}$ converge to $0$, we have $\bm{x_f}$ tends to $ \left(-{L}_1^{-1} {L}_2 \otimes I_n\right)\bm{x_l}$ and $\bm{w_f}$ converges to 0, i.e., $\bm{x_f}(t) \rightarrow \bm{\omega_x}(t)$ and $\bm{w_f}(t) \rightarrow 0$ as $t \rightarrow \infty$. Consequently, state containment control is achieved.

Meanwhile, in the context of distributed control of multi-agent systems, we are interested in the differences between the output values of the leaders and the followers. Therefore, we introduce the performance output variable $\epsilon_i = \sum_{j=1}^N a_{ij} (z_i - z_j), i \in \mathcal{F}$ that captures the differences among the agents according to the communication graph.
Denote $\bm{\epsilon} =\left[\epsilon_1^{\top}, \ldots, \epsilon_M^{\top}\right]^{\top}$, we then have
\begin{equation*}
	\bm{\epsilon} = (L_1 \otimes I_p) \boldsymbol{z_f} + (L_2 \otimes I_p) \boldsymbol{z_l}.
\end{equation*} 
Thus, the performance output variable $\bm{\epsilon}$ reflects the output disagreements between the leaders and followers.
By using the new state variables $\bm{\xi_x}$, $\bm{\xi_w}$ and  performance output variable $\bm{\epsilon}$, the dynamics of the error system can be written in compact form as 
\begin{equation*}
	\begin{aligned}  
		\bm{\dot{\xi}_x} & = 	(L_1 \otimes I_n) 	\bm{\dot{x}_f} + (L_2 \otimes I_n) 	\bm{\dot{x}_l},\\
		\bm{\dot{\xi}_w} & = 	(L_1 \otimes I_n) 	\bm{\dot{w}_f},\\
		\bm{\epsilon} &= (L_1 \otimes I_p) 	\boldsymbol{z_f} + (L_2 \otimes I_p) 	\boldsymbol{z_l}.
	\end{aligned}
\end{equation*}
By interconnecting the  agents \eqref{compact_systems_dynamic} using the protocol \eqref{containment_dynamicprotocol_compact},  we obtain the following  controlled error system
\begin{align}
\begin{bmatrix}\label{compact_containment_dynamic}
		\bm{\dot{\xi}_x} \\
		\bm{\dot{\xi}_w}
	\end{bmatrix} 
	&=
	\begin{bmatrix}
		I_{M} \otimes A & I_{M} \otimes BF \\
			-{L_1} \otimes GC_1 & I_{M} \otimes (A+GC_1) + L_1 \otimes B F
	\end{bmatrix}
	\begin{bmatrix}
		\bm{\xi_x}\\
		\bm{\xi_w}
	\end{bmatrix} \nonumber\\ &\qquad +
	\begin{bmatrix}
		L_1 \otimes E \\
		-{L_1}^2 \otimes GD_1
	\end{bmatrix} \bm{d}, \nonumber\\
	\bm{\epsilon} &= 
	\begin{bmatrix}
		I_{M} \otimes C_2 & I_{M}\otimes D_2 F	
	\end{bmatrix}
	\begin{bmatrix}
		\bm{\xi_x}\\
		\bm{\xi_w}
	\end{bmatrix}.
\end{align}
Denote ${A}_o  =	
		\begin{bmatrix}
			I_{M} \otimes A & I_{M} \otimes BF \\
			-{L_1} \otimes GC_1 & I_{M} \otimes (A+GC_1) + L_1 \otimes B F
		\end{bmatrix},
		\\
		{C}_o = 	\begin{bmatrix}
			I_{M} \otimes C_2 & I_{M} \otimes D_2 F	
		\end{bmatrix},  
		{E}_o   = 	
		\begin{bmatrix}
			L_1 \otimes E \\
			-{L_1}^2 \otimes GD_1
		\end{bmatrix}.$
The impulse response matrix of  the controlled error system \eqref{compact_containment_dynamic} from the external disturbance $\bm{d}$ to the performance output $\bm{\epsilon}$ equals
\begin{equation} \label{containnment_dynamic_impresponse}
	T_{F, G}(t) = {C}_o e^{{A}_ot}{E}_o.
\end{equation}
Subsequently, the associated $H_2$ cost functional  is defined as  
\begin{equation}\label{contaimnet_cost_FG}
	J(F, G) := \int_{0}^{\infty} \text{tr}\left[ T_{F, G}^{\top}(t) T_{F, G}(t)\right] dt,
\end{equation}
which measures the performance of the system 
 \eqref{compact_containment_dynamic} as the square of the $\mathcal{L}_2$-norm of its impulse response matrix. Since the communication among the agents is constrained, the  $H_2$ optimal state containment control problem is non-convex, and it is unknown whether a closed-form solution exists. As an alternative, we solve a suboptimality version of this problem.

\begin{defn}
	The dynamic  protocol \eqref{containment_dynamicprotocol} is said to achieve $H_2$ suboptimal state containment control  for the homogeneous multi-agent system \eqref{containment_leadersdynamic} and \eqref{containment_followersdynamic} if, 
	\begin{itemize}
		\item whenever the external disturbances of the followers are equal to zero, i.e., $\bm{d} = 0$, we have $\bm{x_f}(t) \rightarrow \boldsymbol{\omega_x}(t)$ and $\bm{w_f} (t) \rightarrow 0$ as $t \rightarrow \infty$.
		\item $J(F, G) < \gamma$, where $\gamma$ is an a priori given upper bound.
	\end{itemize}
\end{defn}

More specifically, the problem that we want to address is the following.
\begin{prob}\label{prob4}
	Let $\gamma > 0$. Design gain matrices $F$ and $G$ such that the associated distributed protocol \eqref{containment_dynamicprotocol} achieves state containment control and $J(F, G) < \gamma$.
\end{prob}

\subsection{Protocol Design}\label{subsec_design_homo}
In this subsection, we solve Problem \ref{prob4} and develop a method to compute suitable gain matrices $F$ and~$G$.

According to Assumption \ref{Assumation 1} and Lemma \ref{lemma8},   $L_1$ is positive definite. As a result, $L_1$ can be diagonalized by the orthogonal matrix $U \in \mathbb{R}^{{M} \times {M}}$, i.e., $U^{\top}L_1U = \Lambda = {\rm diag}(\lambda_1, \dots, \lambda_{M}),$
where $0 < \lambda_1 \leq \cdots \leq \lambda_M $ are the eigenvalues of  $L_1$. 

To simplify the problem, we now  show that Problem~\ref{prob4} can be equivalently recast into a number of $H_2$ suboptimal control problems for a set of independent systems. 
To this end, we introduce the following state transformation
\begin{equation} \label{state_transformation1}
	\begin{aligned}
		\begin{bmatrix}
			\bm{\hat{\xi}_x} \\
			\bm{\hat{\xi}_w} 
		\end{bmatrix} 
		&=
		\begin{bmatrix}
			U^{\top} \otimes I_n & 0 \\
			0 & U^{\top} \otimes I_n	
		\end{bmatrix}
		\begin{bmatrix}
			\bm{\xi_x} \\
			\bm{\xi_w}
		\end{bmatrix}.
	\end{aligned}
\end{equation}
Using the transformation \eqref{state_transformation1}, the controlled error system  \eqref{compact_containment_dynamic} becomes 
	\begin{align} \label{compact_containment-transform}
		\begin{bmatrix}
			\bm{\dot{\hat{\xi}}_x} \\
			\bm{\dot{\hat{\xi}}_w}
		\end{bmatrix} 
		&=
		\begin{bmatrix}
			I_{M} \otimes A & I_{M} \otimes BF \\
			-{\Lambda} \otimes GC_1 & I_{M} \otimes (A+G C_1) + \Lambda \otimes B F
		\end{bmatrix} 
		\begin{bmatrix}
		\bm{\hat{\xi}_x}\\
			\bm{\hat{\xi}_w} 
		\end{bmatrix} \nonumber \\&\qquad 
  +
		\begin{bmatrix}
			U^\top L_1 \otimes E \\
			-U^\top L_1^2 \otimes GD_1
		\end{bmatrix} {\bm{d}}
		,
		\\
		\bm{\epsilon}&= 
		\begin{bmatrix}
			U \otimes C_2 &  U \otimes D_2 F
		\end{bmatrix}
		\begin{bmatrix}
			\bm{\hat{\xi}_x}\\
			\bm{\hat{\xi}_w} 
		\end{bmatrix}.\nonumber
	\end{align}
 Denote $\hat{A}_o= 	
			\begin{bmatrix}
				I_{M} \otimes A & I_{M} \otimes BF \\
				-{\Lambda} \otimes GC_1 & I_{M} \otimes (A +GC_1) + \Lambda \otimes B F
			\end{bmatrix}, 
			\\
            \hat{E}_o= 	
			\begin{bmatrix}
				U^\top L_1 \otimes E \\
			-U^\top L_1^2 \otimes GD_1
			\end{bmatrix},
			\hat{C}_o= 
			\begin{bmatrix}
				U \otimes C_2 &  U \otimes D_2 F
			\end{bmatrix}.$
Obviously, the impulse response matrix of the system \eqref{compact_containment-transform} from the disturbance input $\bm{d}$ to the output $\bm{\epsilon}$ is equal to the impulse response matrix \eqref{containnment_dynamic_impresponse}.

In order to proceed, we introduce the following $M$ auxiliary linear systems
\begin{equation*}
	\begin{aligned} 
		\dot{\widetilde{\xi}}_i\left(t\right) & = A \widetilde{\xi}_i\left(t\right)  + {\lambda_i}B \widetilde{u}_i\left(t\right)  +  \frac{1}{\lambda_1} E \widetilde{d}_i\left(t\right),\\
		\widetilde{\zeta}_i\left(t\right) & = C_1 \widetilde{\xi}_i\left(t\right)  +  D_1 \widetilde{d}_i\left(t\right),\\
		\widetilde{\epsilon}_i\left(t\right) & = \lambda_i^2 C_2 \widetilde{\xi}_i\left(t\right)  + \lambda_i^2 D_2 \widetilde{u}_i\left(t\right),
	\end{aligned} i = 1,\dots, M,
\end{equation*}
where
$\lambda_i > 0, i = 1,\dots, M$ are the eigenvalues of the $L_1$, $\lambda_{1}$ is the smallest eigenvalue of $L_1$, and 
$\widetilde{\xi}_i \in \mathbb{R}^n$, $\widetilde{u}_i \in \mathbb{R}^m$, $\widetilde{d}_i \in \mathbb{R}^q$, $\widetilde{\zeta}_i  \in \mathbb{R}^r$, $\widetilde{\epsilon}_i \in \mathbb{R}^p$ are, respectively, the state, the coupling input, the external disturbance, the measured output, and the output to be controlled of the $i$th auxiliary system. By using the  dynamic feedback controllers
\begin{equation}\label{auxiliary_dynamicprotocol_homo}
	\begin{aligned}
		\dot{\widetilde{w}}_i & =A \widetilde{w}_i + \lambda_i  B  \widetilde{u}_i   +  G(C_1\widetilde{w}_i - \widetilde{\zeta}_i) ,\\
		\widetilde{u}_i & = F \widetilde{w}_i,\quad i = 1,\dots, M,
	\end{aligned}
\end{equation}
the controlled closed-loop auxiliary systems can be written as
	\begin{align} \label{decoupled_containment_dynamic_home}
		\begin{bmatrix}
			\dot{\widetilde{\xi}}_i \\
			\dot{\widetilde{w}}_i
		\end{bmatrix} 
		&=
		\begin{bmatrix}
			A &  \lambda_iBF \\
	- GC_1 & A + G C_1 + \lambda_i BF
		\end{bmatrix} 
		\begin{bmatrix}
			\widetilde{\xi}_i\\
			\widetilde{w}_i
		\end{bmatrix} \nonumber \\ & \qquad +
		\begin{bmatrix}
			\frac{1}{\lambda_1}E \\
	-G D_1
		\end{bmatrix} \widetilde{d}_i
		,
		\\
		\widetilde{\epsilon}_i & = 
		\begin{bmatrix}
			\lambda_i^2 C_2 &  \lambda_i^2 D_2 F
		\end{bmatrix}
		\begin{bmatrix}
			\widetilde{\xi}_i\\
			\widetilde{w}_i
		\end{bmatrix},\quad i = 1,\dots, M \nonumber.
	\end{align}
Denote $\widetilde{A}_i = \begin{bmatrix}
	A &  \lambda_iBF \\
	- GC_1 & A + G C_1 + \lambda_i BF 
\end{bmatrix} $, $\widetilde{E}_i = \begin{bmatrix}
	\frac{1}{\lambda_1}E \\
	-G D_1
\end{bmatrix}$, $ \widetilde{C}_i= \begin{bmatrix}
	\lambda_i^2 C_2 &  \lambda_i^2 D_2 F
\end{bmatrix}$, the impulse response matrix of the system \eqref{decoupled_containment_dynamic_home} from the disturbance $\widetilde{d}_i$ to the output $\widetilde{\epsilon}_i$ is 
$\widetilde{T}_{i, F, G}(t) = \widetilde{C}_i e^{\widetilde{A}_it}\widetilde{E}_i$ and the associated $H_2$ cost functional is given by		$J_i(F, G) := \int_{0}^{\infty} \text{tr}\left[ \widetilde{T}_{i, F, G}^{\top}(t) \widetilde{T}_{i, F, G}(t)\right] dt,$ $ i = 1,\dots, M.$

The following theorem  shows that $H_2$ suboptimal state containment control is achieved if $H_2$ suboptimal
control problems of $M$ auxiliary systems are solved.
\begin{thm}\label{thm7}
	Let $F \in \mathbb{R}^{m \times n}$, $G \in \mathbb{R}^{n \times r}$ and $\gamma > 0$ be given. The dynamic protocol \eqref{containment_dynamicprotocol} with   gains $F$ and $G$ achieves state containment control for the multi-agent system  \eqref{containment_leadersdynamic} and \eqref{containment_followersdynamic} if and only if the controllers  \eqref{auxiliary_dynamicprotocol_homo} with $F$ and $G$ internally stabilize  the  closed-loop systems \eqref{decoupled_containment_dynamic_home} for $i=1, \ldots, M$. Moreover, $J(F, G) < \gamma$ if $\sum_{i = 1}^{M} J_i(F, G) < \gamma$.
\end{thm}
\begin{proof}
 We first show that containment control of multi-agent system  \eqref{containment_leadersdynamic} and \eqref{containment_followersdynamic} is achieved if and only if the controllers  \eqref{auxiliary_dynamicprotocol_homo} with $F$ and $ G$ internally stabilize  the  closed-loop systems \eqref{decoupled_containment_dynamic_home} for $i=1, \ldots, M$. Note that the state containment control is achieved if and only if the matrix  $\hat{A}_o$ of the controlled error system \eqref{compact_containment-transform} is Hurwitz. Note also that the matrix  $\hat{A}_o$ is Hurwitz if and only if matrices $\hat{A}_{oi}: = \begin{bmatrix}
		A & BF \\
				-\lambda_i GC_1 & A +GC_1 + \lambda_i B F
	\end{bmatrix} $ are Hurwitz. Now, by applying the state transformation 
	$\begin{bmatrix}
		\hat{\xi}_{ei} \\
		\hat{\xi}_{wi}
	\end{bmatrix} 
	=
	\begin{bmatrix}
		I_n & -\frac{1}{\lambda_i}\\
            0  & I_n
	\end{bmatrix}
	\begin{bmatrix}
		\hat{\xi}_{xi} \\
		\hat{\xi}_{wi}
	\end{bmatrix}$ on the controlled error system \eqref{compact_containment-transform}, the matrices $\hat{A}_{oi}$ can be transformed into $\hat{A}_{ei}: = \begin{bmatrix}
		A + GC_1 & 0 \\
				-\lambda_i BF & A + \lambda_i B F
	\end{bmatrix}$.  It is then easy to see that the controlled error system \eqref{compact_containment-transform} is internally stable, and subsequently, containment control of multi-agent system \eqref{containment_leadersdynamic} and \eqref{containment_followersdynamic} is achieved if and only if the matrices $A+\lambda_i BF$ and $A +G C_1$ for $i = 1,\dots, M$ are Hurwitz. 
 
	Now, by applying the state transformation 
	$\begin{bmatrix}
		\widetilde{w}_i \\
		\widetilde{e}_i
	\end{bmatrix} 
	=
	\begin{bmatrix}
		0  & I_n \\
		I_n & -I_n	
	\end{bmatrix}
	\begin{bmatrix}
		\widetilde{\xi}_i \\
		\widetilde{w}_i
	\end{bmatrix}$ on the auxiliary systems \eqref{decoupled_containment_dynamic_home},  the matrices $\widetilde{A}_i = \begin{bmatrix}
		A &  \lambda_iBF \\
	- GC_1 & A + G C_1 + \lambda_i BF 
	\end{bmatrix} $ can be transformed into $\widetilde{A}_{ei}: = \begin{bmatrix}
		A + \lambda_i BF&   -GC_1 \\
		0 & A + G C_1 
	\end{bmatrix}$. It is easy to see that all the $M$ closed-loop systems  \eqref{decoupled_containment_dynamic_home} are internally stable if and only if the matrices $A+\lambda_i BF$ and $A +G C_1$ are Hurwitz. It then follows from the above that state containment control for the multi-agent system \eqref{containment_leadersdynamic} and \eqref{containment_followersdynamic} is achieved if and only if the $M$ auxiliary closed-loop systems are internally stabilized by the controllers  \eqref{auxiliary_dynamicprotocol_homo}.
	
	Next, let $F$ and $G$ be such that matrices $A+\lambda_i BF$ and $A + G C_1$ are Hurwitz for $i = 1,\dots, M$. We show that $\sum_{i = 1}^{M} J_i(F, G) < \gamma$ implies $J(F, G) < \gamma$. Recall that $J_i(F, G) 
			= \int_{0}^{\infty} {\rm tr} (\widetilde{T}_{i, F, G}(t)^\top \widetilde{T}_{i, F, G}(t)) dt$ and $\widetilde{T}_{i, F, G}(t) = \widetilde{C}_i e^{\widetilde{A}_it}\widetilde{E}_i$, then it holds that
 \begin{equation}\label{J_costsum}
		\begin{aligned}
			\sum_{i = 1}^{M} &J_i(F, G) 
			 =  \int_{0}^{\infty} \sum_{i = 1}^{M}\text{tr}\left[ (\widetilde{C}_{i} e^{\widetilde{A}_{i} t}\widetilde{E}_{i})^\top(\widetilde{C}_{i} e^{\widetilde{A}_{i}t}\widetilde{E}_{i})\right] dt\\
            = &  \int_{0}^{\infty} \sum_{i = 1}^{M} \text{tr}\left[ (\bar{C}_{i} e^{\bar{A}_{i} t}\bar{E}_{i})^\top(\frac{\lambda_i}{\lambda_1})^\top(\frac{\lambda_i}{\lambda_1})(\bar{C}_{i} e^{\bar{A}_{i}t}\bar{E}_{i})\right]dt\\
            \geq & \int_{0}^{\infty} \sum_{i = 1}^{M} \text{tr}\left[ (\bar{C}_{i} e^{\bar{A}_{i} t}\bar{E}_{i})^\top(\bar{C}_{i} e^{\bar{A}_{i}t}\bar{E}_{i})\right]dt
		\end{aligned}
	\end{equation}
 where $\bar{C}_{i}: = \begin{bmatrix}
				C_2 & D_2 F
			\end{bmatrix} $, $\bar{E}_{i}: = \begin{bmatrix}
				\lambda_i E \\
			-\lambda_i^2  GD_1
			\end{bmatrix},$ and $\bar{A}_{i}: = \begin{bmatrix}
		A & \lambda_i BF \\
				- GC_1 & A +GC_1 + \lambda_i B F
	\end{bmatrix}$ for $i = 1,\dots, M$, and in the last step of \eqref{J_costsum} we have used $0< \lambda_1 \leq \lambda_i$  for  $i = 1,\dots, M$. In fact, it is not difficult to see that
  \begin{equation}\label{J_cost}
		\begin{aligned}
  & \int_{0}^{\infty} \sum_{i = 1}^{M} \text{tr}\left[ (\bar{C}_{i} e^{\bar{A}_{i} t}\bar{E}_{i})^\top(\bar{C}_{i} e^{\bar{A}_{i}t}\bar{E}_{i})\right]dt\\ 
=   &  \int_{0}^{\infty}\text{tr}\left[ (\hat{C}_o e^{\hat{A}_o t}\hat{E}_o)^\top(\hat{C}_o e^{\hat{A}_ot}\hat{E}_o)\right]dt =  J(F ,G).
		\end{aligned}
	\end{equation}
It then follows from \eqref{J_costsum} and \eqref{J_cost} that
 \begin{equation*}
		\begin{aligned}
			J(F, G)  
			\leq &  \sum_{i = 1}^{M} J_i(F, G) < \gamma.
		\end{aligned}
	\end{equation*}
	This completes the proof.
\end{proof}

By applying Theorem \ref{thm7}, the   $H_2$ suboptimal state containment control problem of multi-agent system \eqref{containment_leadersdynamic} and~\eqref{containment_followersdynamic} can be converted into a number of $H_2$ suboptimal control problems for the $M$ independent auxiliary systems \eqref{decoupled_containment_dynamic_home} by dynamic controllers~\eqref{auxiliary_dynamicprotocol_homo}.
Furthermore,   the  following lemma
provides conditions under which  the $M$ closed-loop systems \eqref{decoupled_containment_dynamic_home} are internally stable   while achieving $\sum_{i = 1}^{M} J_i(F, G) < \gamma$.
\begin{lem}\label{lemma9}
	Let $\gamma > 0$ be given. Assume that $D_1E^{\top}  =0$, $D_2^{\top} C_2 =0$, $D_1 D_1^{\top} =I_r$, and ${D}_2^\top {D}_2 =I_m$. The dynamic controllers \eqref{auxiliary_dynamicprotocol_homo} with gain matrices $F$ and $G=-QC_1^\top$ internally stabilize all $M$ closed-loop systems \eqref{decoupled_containment_dynamic_home} and $\sum_{i = 1}^{M} J_i(F, G) < \gamma$ if there exist $P_i > 0$, $i = 1,\dots, M,$  and $Q > 0$ satisfying
	\begin{align}
		& (A + \lambda_i B F )^{\top} P_i  + P_i (A + \lambda_i B F ) \nonumber\\ &\qquad\quad+ (\lambda_i^2 C_2+ \lambda_i^2  D_2 F )^\top (\lambda_i^2 C_2+ \lambda_i^2  D_2 F) < 0, \label{Lemma9_P} \\
		& A Q + Q A^\top - Q C_1^\top C_1 Q +\frac{1}{\lambda_1^2} E E^\top  <0, \label{Lemma9_Q} \\
		& \sum_{i=1}^{M}
		{\rm tr} \left(C_1 Q P_i QC_1^\top  \right) + \lambda_i^4 {\rm tr} \left(  C_2 Q  C_2^\top  \right)  < \gamma \label{Lemma9_gamma}.
	\end{align}		
\end{lem}
\begin{proof} By \eqref{Lemma9_gamma}, for $\epsilon_i >0$  sufficiently small, we have $\sum_{i=1}^{M} \gamma_i <\gamma$, where $\gamma_i :=
		{\rm tr} \left(C_1 Q P_i QC_1^\top  \right) + \lambda_i^4 {\rm tr} \left(  C_2 Q  C_2^\top  \right)  + \epsilon_i$.  By taking $\bar{A} = A$, $\bar{B} = \lambda_i B$, $\bar{C}_1 =  C_1$, $\bar{D}_1 = D_1$, $\bar{C}_2 =  \lambda_i^2 C_2$, $\bar{D}_2 = \lambda_i^2 D_2$, and $\bar{E} =  \frac{1}{\lambda_1} E$ in Lemma \ref{lem_1}, there exist $P_i > 0$ and $Q  > 0$ such that \eqref{Lemma9_P}, \eqref{Lemma9_Q} and ${\rm tr} \left(C_1 Q P_i QC_1^\top  \right) + \lambda_i^4 {\rm tr} \left(  C_2 Q  C_2^\top  \right) < \gamma_i$ hold for all $i = 1,\dots, M$. Additionally, all closed-loop systems \eqref{decoupled_containment_dynamic_home} are internally stable and $J_i(F, G) < \gamma_i$ by dynamic controllers \eqref{auxiliary_dynamicprotocol_homo} with gains matrices $F$ and $G=-QC_1^\top$. Subsequently, $\sum_{i=1}^{M} J_i(F, G) < \sum_{i=1}^{M} \gamma_i < \gamma $.
\end{proof}

Note that the assumptions $D_1 D_1^{\top} =I_r$ and ${D}_2^\top {D}_2 =I_m$  in Lemma \ref{lemma9} are made  to simplify notation and can be easily relaxed to the regularity conditions $D_1 D_1^{\top} > 0$ and $D_2^{\top}D_2 > 0$ as in Lemma \ref{lem_1}. Nevertheless, Lemma \ref{lemma9} does not yet provide any method to compute a suitable gain matrix $F$. Therefore, combined with the $G$ given above, a design method for obtaining an appropriate distributed dynamic protocol \eqref{containment_dynamicprotocol}  is established below.
\begin{thm}\label{thm8}
	Let $\gamma > 0$ be given. Assume that $D_1E^{\top}  =0$, $D_2^{\top} C_2 =0$, $D_1 D_1^{\top} =I_r$, and ${D}_2^\top {D}_2 =I_m$. 
	We consider two cases for the parameter $c_p$:
	\begin{enumerate}
		\item \label{thm8_pcase1}
        if $0 < {c_p} < \frac{2}{(\lambda_1+\lambda_M)(\lambda_1^2+\lambda_{M}^2)}$,
        where $\lambda_1$ is the smallest eigenvalue  and $\lambda_{M}$ is the largest eigenvalue of $L_1$.
		Then there exists $P>0$ satisfying
		\begin{equation}\label{thm8_c2_p}
			A^{\top} P + P A + ({c_p}^2 \lambda_{1}^4 - 2 {c_p} \lambda_1) PB B^{\top} P +\lambda_M^4 {C_2}^{\top} {C_2} <0.
		\end{equation}
		\item \label{thm8_pcase2}
	   if $\frac{2}{(\lambda_1+\lambda_M)(\lambda_1^2+\lambda_{M}^2)} \leq {c_p} <\frac{2}{\lambda_{M}^3},$
		then there exists $P>0$ satisfying
		\begin{equation}\label{thm8_c1_p}
			A^{\top} P + P A + ({c_p}^2 \lambda_{M}^4 - 2 {c_p} \lambda_M) PB B^{\top} P +\lambda_M^4 {C_2}^{\top} {C_2} <0.
		\end{equation}
	\end{enumerate}
 
	Let $Q>0$ satisfy the following Riccati inequality
		\begin{equation}\label{thm8_c1_q}
			A Q + Q A^{\top} - QC_1^\top C_1Q + \frac{1}{\lambda_1^2}EE^{\top} <0.
		\end{equation}
If, in addition, $P$ and $Q$ also satisfy
	\begin{equation}\label{thm8_gamma}
		{\rm tr}\left(C_1 Q P QC_1^\top  \right) + \lambda_M^4  {\rm tr} \left(  C_2 Q  C_2^\top  \right) <  \frac{\gamma}{M},
	\end{equation}
	then the protocol \eqref{containment_dynamicprotocol} with $F := -{c_p}B^\top P$ and $G := -QC_1^\top$ achieves state containment control for multi-agent system \eqref{containment_leadersdynamic} and \eqref{containment_followersdynamic}, and the protocol is suboptimal, i.e., $J(F, G) < \gamma$.
	
	\begin{proof}
            First, note that \eqref{thm8_c1_q} is exactly \eqref{Lemma9_Q}.
		For case \ref{thm8_pcase2}) above, it follows from $\frac{2}{(\lambda_1+\lambda_M)(\lambda_1^2+\lambda_{M}^2)} \leq {c_p} <\frac{2}{\lambda_{M}^3}$ that 
		${c_p}^2 \lambda_{M}^4 - 2 {c_p} \lambda_{M} <0$ and that the Riccati inequality \eqref{thm8_c1_p} has a positive definite solution $P$. Since ${c_p}^2 \lambda_1^4 - 2 {c_p} \lambda_1 \leq {c_p}^2 \lambda_i^4 - 2 {c_p} \lambda_i \leq {c_p}^2 \lambda_{M}^4 - 2 {c_p} \lambda_{M} <0$ and $\lambda_i \leq \lambda_{M}$ for  $i = 1,\dots, M$, the positive definite solution $P$ of \eqref{thm8_c1_p} also satisfies the  Riccati inequalities \begin{equation}\label{thm8_c1_pf3}
			A^{\top} P + P A + ({c_p}^2 \lambda_i^4 - 2 {c_p} \lambda_i) PB B^{\top} P +\lambda_i^4 {C_2}^{\top} {C_2} <0
		\end{equation}
  for $i=1,\ldots,M$. 
		Recall the conditions ${D_2}^{\top} {C_2} =0$ and ${D_2}^\top {D_2} =I_m$, then \eqref{thm8_c1_pf3} immediately yields 
		\begin{equation}\label{thm8_c1_pf1}
			\begin{aligned}
				&(A - {c_p}\lambda_i B B^{\top}P)^{\top}P + P(A- {c_p}\lambda_iB B^{\top}P) +\\
				&\quad (\lambda_i^2C_2 - {c_p}\lambda_i^2 {D_2} B^\top P)^\top  (\lambda_i^2 C_2 - {c_p}\lambda_i^2{D_2} B^\top P) < 0.
			\end{aligned}
		\end{equation}	
By taking $P_i = P$ and $F=-c_pB^\top P$, it follows that 
  \eqref{Lemma9_P} holds.
		Furthermore,  \eqref{thm8_gamma} implies that also \eqref{Lemma9_gamma}  holds.
		Then it follows    from Lemma~\ref{lemma9} that all $M$ closed-loop systems \eqref{decoupled_containment_dynamic_home} are internally stable and $\sum_{i=1}^{M} J_i(F, G)<\gamma$. 
    
        Subsequently, it follows from Theorem \ref{thm7} that the dynamic feedback protocol~\eqref{containment_dynamicprotocol} achieves containment control for the multi-agent system \eqref{containment_leadersdynamic}, \eqref{containment_followersdynamic} and $J(F, G)<\gamma$.
        
         Similarly, for case \ref{thm8_pcase1}), it can be verified that  ${c_p}^2 \lambda_{M}^4 - 2 {c_p} \lambda_{M} \leq {c_p}^2 \lambda_i^4 - 2 {c_p} \lambda_i \leq {c_p}^2 \lambda_1^4 - 2 {c_p} \lambda_1 <0$ and $\lambda_i \leq \lambda_{M}$ for  $i = 1,\dots, M$. The detailed proof is similar to case \ref{thm8_pcase2}) and thus is omitted here. 
	\end{proof}
\end{thm}
\begin{rem}\label{remark4}
	Theorem \ref{thm8} states that by choosing feasible $c_p$, $P$ and $Q$, the distributed protocol~\eqref{containment_dynamicprotocol} with gains $F = -{c_p}B^\top P$ and $G = -Q{C_1}^\top$ is suboptimal. Then the question arises: how do we find the smallest upper bound $\gamma$ such that $ {\rm tr}\left(C_1 Q P QC_1^\top  \right) + \lambda_M^4  {\rm tr} \left(  C_2 Q  C_2^\top  \right) <  \frac{\gamma}{M}$?
	
	It is easy to note that in general the smaller $P$ and $Q$ lead to smaller ${\rm tr} \left(C_1 Q P QC_1^\top  \right)$ and ${\rm tr} \left(  C_2 Q  C_2^\top  \right)$, consequently, smaller $\gamma$. To find the feasible selection of $\gamma$, we could try to find $P$ and $Q$ as small as possible. Indeed, we can compute $P$ satisfying \eqref{thm8_c2_p} or \eqref{thm8_c1_p} by finding $P({c_p},\delta) > 0$ satisfying
	\begin{equation*}\label{c1p}
		A^{\top} P + P A - r_{p1} PB B^{\top} P +\lambda_M^4 {C_2}^{\top} C_2  + \delta I_n = 0, 
	\end{equation*}
	\begin{equation*}\label{c2p}
		A^{\top} P + P A - r_{p2} PB B^{\top} P +\lambda_M^4 {C_2}^{\top} C_2 + \delta I_n = 0,
	\end{equation*}
	where $\delta > 0 $,  $ r_{p1} = (-{c_p}^2 \lambda_1^4 + 2 {c_p} \lambda_1)$ and $r_{p2} = (-{c_p}^2 \lambda_{M}^4 + 2 {c_p} \lambda_{M})$. Obviously, the larger the coefficient $r_{p1}$ (or $r_{p2})$ and the smaller $\delta$ lead to the smaller $P$.
	It can then be shown that, for both cases, a small  $P$ can be computed by taking ${c_p}^* = \frac{2}{(\lambda_1+\lambda_M)(\lambda_1^2+\lambda_{M}^2)}$, which leads to the biggest   $r_{p1}$ and $r_{p2}$, and  $\delta > 0$ very close to $0$.
	Similarly, a small $Q > 0$ satisfying \eqref{thm8_c1_q}  can be computed by solving $Q(\eta) > 0 $ with $\eta > 0 $ that satisfies
	\begin{equation*}\label{c1q}
		A Q + Q A^{\top} - QC_1^\top C_1Q + \frac{1}{\lambda_1^2}EE^{\top}+\eta I_n = 0.
	\end{equation*}	
    Therefore, if we choose $\eta > 0$ and $\delta > 0$ very close to 0, and choose ${c_p} = \frac{2}{(\lambda_1+\lambda_M)(\lambda_1^2+\lambda_{M}^2)}$, we find the `best' small $\gamma$ in the sense as explained above.
\end{rem}

\section{\texorpdfstring{${H}_2$}~ suboptimal output containment control of heterogeneous linear multi-agent systems}
In this section, we consider the ${H}_2$ suboptimal output containment control problem for heterogeneous multi-agent systems. In Subsection \ref{subsec_problem_heter}, we formulate this problem. Then we propose a design method in Subsection \ref{subsec_design_heter} to develop a  distributed   protocol by dynamic output feedback that solves this problem.

\label{sec_heterogenous}

\subsection{Problem Formulation}\label{subsec_problem_heter}
We consider a heterogeneous linear multi-agent system with  $N$ agents that consists of $M$   followers subjected to external disturbances and $N-M$ autonomous leaders. 
Without loss of generality, we assume that the agents labeled $1$ through $M$ are the followers, which we refer to as the follower set $\mathcal{F} \overset{\Delta}{=} \left\{1, \ldots, M\right\}$, while the agents labeled $M+1$ through $N$ are the leaders, referred to as the leader set $\mathcal{L} \overset{\Delta}{=} \left\{M + 1,\dots, N\right\}$.
The dynamics of the $i$th leader is represented by 
\begin{equation}\label{hleaders_dynamic}
		\dot{x}_i  = S x_i,\quad z_i = R x_i,\qquad i \in  \mathcal{L},
\end{equation}
where $x_i \in \mathbb{R}^r$ and $z_i \in \mathbb{R}^p$ are, respectively, the state and the output to be controlled. The matrices $S$ and $R$ are of compatible dimensions.
The dynamics of the $i$th follower is represented by
\begin{equation}\label{hfollowers_dynamic}
	\begin{aligned} 
		\dot{x}_i & = A_i x_i  + B_i u_i  + E_i d_i,\\
            y_i & = C_{1i} x_i  + D_{1i} d_i,\\
		z_i & = C_{2i} x_i  + D_{2i} u_i,
	\end{aligned}\qquad i \in \mathcal{F},
\end{equation}  
where $x_i \in \mathbb{R}^{n_i}$, $u_i \in \mathbb{R}^{m_i}$, $d_i \in \mathbb{R}^{q_i}$, $y_i \in \mathbb{R}^{r_i}$, and $z_i \in \mathbb{R}^p$ are, respectively, the state, the coupling input, the unknown external disturbance, the measured output, and the output to be controlled. 
The matrices $A_i$, $B_i$, $C_{1i}$, $C_{2i}$, $D_{1i}$, $D_{2i}$, and $E_i$ are of compatible dimensions. 
Throughout this section, we assume that the pairs $\left(A_i, B_i\right)$ are stabilizable and the pairs $\left(C_{1i}, A_i\right)$ are detectable. The followers are heterogeneous agents in the sense that the system dynamics and state dimensions among them are generally different. In this regard, the control objective, therefore, is to ensure that the outputs of the followers converge to the  convex hull spanned by the outputs of the leaders  \cite{zuo2017output,qin2018output}.

\begin{rem}
    Note that the multi-agent system considered in this paper is more general than that considered in~\cite{qin2018output}. Indeed,  the  followers \eqref{hfollowers_dynamic}   contain measured outputs $y_i$ and controlled outputs $z_i$ with a feedthrough term, whereas the followers in \cite{qin2018output} contain only controlled outputs without feedthrough terms. 
\end{rem}

It has been shown in \cite{qin2018output} that  the solvability of certain regulator equations derived from the internal model principle is necessary for the output containment control of heterogeneous linear multi-agent systems. Accordingly, throughout this section,  we assume that there exists a positive integer $r$ such that the regulator equations
\begin{equation}
\begin{aligned}\label{regulation}
  &A_i \Pi_i  + B_i \Gamma_i  = \Pi_i S,\\
  &C_{2i} \Pi_i  + D_{2i} \Gamma_i  = R,  \qquad i \in \mathcal{F},
\end{aligned}
\end{equation}
have solutions $\Pi_i \in \mathbb{R}^{n_i\times r}$, $\Gamma_i \in \mathbb{R}^{m_i\times r}$, where we assume that the pair $\left(R, S\right)$ is observable and the eigenvalues of $S$ lie on the imaginary axis.

Following \cite{qin2018output}, we consider that the leaders \eqref{hleaders_dynamic} and the followers \eqref{hfollowers_dynamic} are interconnected by a distributed observed-based   dynamic output feedback protocol of the form 
\begin{equation}\label{heter_dynamicprotocol}
	\begin{aligned}
            \dot{w}_i & = A w_i + B u_i + G_i\left(C_{1i}w_i - y_i \right),\\
            \dot{v}_i & = S v_i + \sum_{j=1}^M a_{ij} (v_j - v_i) + \sum_{j={M+1}}^N a_{ij} (x_j - v_i),\\
		u_i & = F_i(w_i - \Pi_i v_i) + \Gamma_i v_i, \qquad i \in \mathcal{F},
	\end{aligned}
\end{equation} 
where $G_i \in \mathbb{R}^{n_i \times r_i}$ and  $F_i \in \mathbb{R}^{m_i \times n_i}$ are local feedback gains to be designed, and $a_{ij}$ is the $ij$th entry of the adjacency matrix $\mathcal{A}$ associated with the communication graph $\mathcal{G}$ which satisfies Assumption \ref{Assumation 1}. Note that in protocol \eqref{heter_dynamicprotocol}, the first equation with   state $w_i \in \mathbb{R}^{n_i}$ is a Luenberger observer to estimate   state $x_i$, the second equation is an auxiliary reference system with the state $v_i \in \mathbb{R}^{r}$, and the third equation is an output regulation controller.

Denote $\bm{x_f} = \left[x_1^{\top}, \ldots, x_M^{\top}\right]^{\top}$,  $\bm{x_l} =\left[x_{M+1}^{\top}, \ldots, x_N^{\top}\right]^{\top}$, \\$\bm{y_f}= \left[y_1^{\top}, \ldots, y_M^{\top}\right]^{\top}$, $\bm{u} =\left[u_1^{\top}, \ldots, u_M^{\top}\right]^{\top}$, $\bm{d} =[d_1^{\top}, \ldots, $\\$d_M^{\top}]^{\top}$, $\bm{w_f} = \left[w_1^{\top}, \ldots, w_M^{\top}\right]^{\top}$, $\bm{v_f} = \left[v_1^{\top}, \ldots, v_M^{\top}\right]^{\top}$, $\bm{z_f} =\\ \left[z_1^{\top}, \ldots, z_M^{\top}\right]^{\top}$, and $\bm{z_l} =\left[z_{M+1}^{\top}, \ldots, z_N^{\top}\right]^{\top}$.  Denote $A = {\rm blockdiag}(A_1, \ldots, A_M)$ and likewise define $B, C_1, C_2, D_1,$\\$ D_2$, and $E$.
We can then write the agents \eqref{hleaders_dynamic} and~\eqref{hfollowers_dynamic} in compact form as
\begin{equation}\label{hcompact_systems}
    \begin{aligned} 
            \dot{\bm{x_l}} &= (I_{N-M} \otimes S) \bm{x_l},\\
        \bm{z_l} &= (I_{N-M} \otimes R)\bm{x_l}, \\
       \bm{\dot{{x}}_f} &= A \bm{x_f} + B \bm{u} + E \bm{d},\\
        \bm{y_f} &= C_1\bm{x_f} + D_1 \bm{d},\\
        \bm{z_f} &= C_2\bm{x_f} + D_2 \bm{u}.
    \end{aligned}
\end{equation} 
Similarly, denote $F = {\rm blockdiag}(F_1, \ldots, F_M)$, and likewise define $G, \Pi$, and $\Gamma$. The protocol \eqref{heter_dynamicprotocol} can be written as 
\begin{equation}\label{heter_dynamicprotocol_compact}
\begin{aligned}
	\bm{\dot{w}_f} & = A \bm{w_f} + B \bm{u} +  G(C_1 \bm{w_f} - \bm{y_f}),\\
 	\bm{\dot{v}_f} & =(I_{M} \otimes S) \bm{v_f} - ({L_{1}}\otimes I_{r})\bm{v_f} - ({L_{2}} \otimes I_r)\bm{x_l},\\
        \bm{u} & = F\bm{w_f} + (\Gamma - F \Pi) \bm{v_f},
\end{aligned}
\end{equation}
where ${L}_1 \in \mathbb{R}^{M\times M}$ and ${L}_2 \in \mathbb{R}^{M\times\left(N-M\right)}$ are defined in \eqref{laplacian}.
Meanwhile, the regulation equations \eqref{regulation} can be written as 
\begin{equation}
\begin{aligned}\label{regulation_compact}
  A \Pi + B \Gamma  = \Pi & (I_M \otimes S),\\
  C_2 \Pi  + D_2 \Gamma  = & I_M \otimes R.
\end{aligned}
\end{equation}
Foremost, we want the protocol \eqref{heter_dynamicprotocol_compact} to achieve the output containment control problem for multi-agent system~\eqref{hcompact_systems} without considering disturbances. 

Here, output containment control means that the outputs of the followers converge into the convex hull  $\boldsymbol{\omega_z}(t)$ formed by the outputs of the leaders, i.e., 
\begin{equation}\label{outputconvex_hull}
\boldsymbol{\omega_z}(t) \overset{\Delta}{=}\left(-{L}_1^{-1} {L}_2 \otimes I_p\right)\bm{z_f}(t)
\end{equation}
where $\bm{z_f}(t)$ is the controlled output of the leaders and recall that the sum of each row  of $-{L}_1^{-1} {L}_2$ equal to $1$. We now formally define output containment control in the following.
\begin{defn}\label{definition4}
   The protocol \eqref{heter_dynamicprotocol_compact} is said to achieve output containment control for the multi-agent system~\eqref{hcompact_systems} if the outputs of the followers $\bm{z_f}$ converge into the convex hull $\boldsymbol{\omega_z}$ formed by the outputs of the leaders $\bm{z_l}$, the protocol state $\bm{w_f}$ converges to the followers' state $\bm{x_f}$, and the state $\bm{v_f}$ of the auxiliary reference system converges into the convex hull formed $\boldsymbol{\omega_v}$ by the state $\bm{x_l}$, i.e., $\bm{z_f}(t) \rightarrow \boldsymbol{\omega_z}(t)$, $\bm{w_f}(t) \rightarrow \bm{x_f}(t)$, and $\bm{v_f}(t) \rightarrow \bm{\omega_v}(t)$ where $\boldsymbol{\omega_v}(t) \overset{\Delta}{=}\left(-{L}_1^{-1} {L}_2 \otimes I_r\right)\bm{x_l}(t)$, as $t \rightarrow \infty$.
\end{defn}

To proceed, we introduce a new state error variable for each auxiliary reference system in protocol~\eqref{heter_dynamicprotocol} as
   $ \xi_{i} = \sum_{j=1}^M a_{ij} (v_i - v_j) + \sum_{j={M+1}}^N a_{ij} (v_i - x_j ), i \in \mathcal{F}.$
Denote $\bm{\xi_{i}} =\left[\xi_{1}^{\top}, \ldots, \xi_{M}^{\top}\right]^{\top}$,
 we then have
\begin{equation*}
	\bm{\xi} = (L_1 \otimes I_r) \bm{v_f} + (L_2 \otimes I_r) \bm{x_l}. 
\end{equation*}
Note that,  whenever $\bm{\xi}$ converges to $0$, we have   $\bm{v_f}$ tends to $ \left(-{L}_1^{-1} {L}_2 \otimes I_r\right)\bm{x_l}$, i.e., $\bm{v_f}(t) \rightarrow \bm{\omega_v}(t)$. Then state containment control of the auxiliary reference system is achieved.

We also introduce two new state error variables for each follower as $e_{i} = x_i - w_j, \delta_i = x_i - \Pi_i v_i,$ for $i \in \mathcal{F}$.
Denote $\bm{e} =\left[{{e_1}} ^{\top}, \ldots, {{e_M}}^{\top}\right]^{\top}$, $\bm{\delta} =\left[\delta_1^{\top}, \ldots, \delta_M^{\top}\right]^{\top}$, we then have
\begin{equation*}
        \bm{e} = \bm{x_f} - \bm{w_f}, ~~
        \bm{\delta} = \bm{x_f} - \Pi \bm{v_f}.
\end{equation*}
Whenever $\bm{e}$ and $\bm{\delta}$ converge to $0$, $\bm{w_f}$ converges to $\bm{x_f}$ and $\bm{x_f}$ converges to $\Pi \bm{v_f}$.

In the context of distributed control of multi-agent systems, we focus on the differences between the output values of the leaders and the followers. Thus, we introduce the performance output variable $\epsilon_i = \sum_{j=1}^N a_{ij} (z_i - z_j), i \in \mathcal{F}$.
Denote $\bm{\epsilon} =\left[\epsilon_1^{\top}, \ldots, \epsilon_M^{\top}\right]^{\top}$,   we have
\begin{equation*}
	\bm{\epsilon} = (L_1 \otimes I_p) \boldsymbol{z_f} + (L_2 \otimes I_p) \boldsymbol{z_l}.
\end{equation*} 
Therefore, the performance output variable $\bm{\epsilon}$ reflects the disagreements between the outputs of the leaders and the followers. Additional, it is worth noting that whenever $\bm{\epsilon}$ converges to $0$, $\bm{z_f}$ tends to $ \left(-{L}_1^{-1} {L}_2 \otimes I_p\right)\bm{z_l}$,  which is exactly the convex hull $\boldsymbol{\omega_z}(t)$ formed by the outputs of leaders $\bm{z_l}$, i.e.,  output containment control for the multi-agent system~\eqref{hcompact_systems} is achieved.

By using the new state variables $\bm{\xi}, \bm{e}, \bm{\delta}$ and the performance output variable~$\bm{\epsilon}$, the dynamics of the error system can be written in compact form as 
\begin{equation*}
	\begin{aligned}  
         \dot{\bm{e}} &= \bm{\dot{x}_f} - \bm{\dot{w}_f},\\
        \dot{\bm{\delta}} &= \bm{\dot{x}_f} - \Pi \bm{\dot{v}_f}\\
		\dot{\bm{\xi}} & = 	(L_1 \otimes I_r) 	\bm{\dot{v}_f} + (L_2 \otimes I_r) 	\dot{\bm{x_l}},\\
		\bm{\epsilon} &= (L_1 \otimes I_p) 	\boldsymbol{z_f} + (L_2 \otimes I_p) 	\boldsymbol{z_l}.
	\end{aligned}
\end{equation*}

By interconnecting  the multi-agent system \eqref{hcompact_systems} using the protocol \eqref{heter_dynamicprotocol_compact} and the regulation equations \eqref{regulation_compact}, the controlled error system satisfies the following dynamics
\begin{align}
	\begin{bmatrix}\label{compact_heter_containment}
		\dot{\bm{e}} \\
		\dot{\bm{\delta}}\\
            \dot{\bm{\xi}}
	\end{bmatrix} 
	&=
	\begin{bmatrix}
		A + G C_1 & 0 & 0 \\
		-BF & A + BF & \Pi \\
            0 & 0 & I_M \otimes S - L_1 \otimes I_r
	\end{bmatrix}
	\begin{bmatrix}
		\bm{e}\\
		\bm{\delta}\\
            \bm{\xi}
	\end{bmatrix}\nonumber \\  
 &\qquad +
	\begin{bmatrix}
		E + G D_1 \\
            E \\
		0
	\end{bmatrix} \bm{d} 
	,
	\\
	\bm{\epsilon}= &
	\begin{bmatrix}
		-(L_1 \otimes I_p) D_2 F & \!\!\!  (L_1 \otimes I_p) (C_2 + D_2 F) & \!\!\!\!  I_M \otimes R	
	\end{bmatrix}
	\begin{bmatrix}
		\bm{e}\\
		\bm{\delta}\\
            \bm{\xi}
	\end{bmatrix}.\nonumber
\end{align}
Denote ${A}_o  =	
		\begin{bmatrix}
		A + G C_1 & 0 & 0 \\
		-BF & A + BF & \Pi \\
            0 & 0 & I_M \otimes S - L_1 \otimes I_r
	\end{bmatrix}$,
		${C}_o = 	\begin{bmatrix}
		-(L_1 \otimes I_p) D_2 F & (L_1 \otimes I_p) (C_2 + D_2 F) & I_M \otimes R	
	\end{bmatrix}$, ${E}_o   = 	
		\begin{bmatrix}
		E + G D_1 \\
            E \\
		0
	\end{bmatrix}$. 
Then the impulse response matrix of the controlled error system~\eqref{compact_heter_containment} from the external disturbance $\bm{d}$ to the performance output $\bm{\epsilon}$ is equal to 
\begin{equation} \label{heter_containment_dynamic}
		T_{F, G}(t) = {C}_o e^{{A}_ot}{E}_o.
\end{equation}
Furthermore, the associated ${H}_2$ cost functional is introduced as  
\begin{equation}\label{heter_cost_FG_containment}
	J(F, G) := \int_{0}^{\infty} \text{tr}\left[ T_{F, G}^{\top}(t) T_{F, G}(t)\right] dt,
\end{equation}
which measures the performance of the system 
 \eqref{compact_heter_containment} as the square of the $\mathcal{L}_2$- norm of its impulse response \eqref{heter_containment_dynamic}.

However, since the agent communication is constrained, this $H_2$ optimal output containment control problem for the multi-agent system \eqref{hleaders_dynamic} and \eqref{hfollowers_dynamic} is non-convex and a closed-form solution is unknown to exist. As an alternative, we solve a version of this problem that requires only suboptimality.
\begin{defn}
	The protocol \eqref{heter_dynamicprotocol} is said to achieve ${H}_2$ suboptimal output
 containment control  for the heterogeneous multi-agent system \eqref{hleaders_dynamic} and \eqref{hfollowers_dynamic} if, 
	\begin{itemize}
		\item whenever the external disturbances of all followers are equal to zero, i.e., $\bm{d} = 0$,  we have 
		$\bm{z_f}(t) \rightarrow \boldsymbol{\omega_z}(t)$, $\bm{w_f}(t) \rightarrow \bm{x_f}(t)$, and $\bm{v_f}(t) \rightarrow \bm{\omega_v}(t)$ as $t \rightarrow \infty$.
		\item $J(F, G) < \gamma$, where $\gamma$ is an a priori given upper bound.
	\end{itemize}
\end{defn}

More concretely, we want to address the following problem.
\begin{prob}\label{prob1}
	Let $\gamma > 0$ be a given tolerance. Design gain matrices $F_i$ and $G_i$ for $i \in \mathcal{F}$  such that the distributed dynamic protocol \eqref{heter_dynamicprotocol} 
	achieves output containment control and   $J(F, G) < \gamma$.
\end{prob}

\subsection{Protocol Design}\label{subsec_design_heter}
In this subsection, we address Problem \ref{prob1}, and a design approach is provided to obtain appropriate gain matrices $F_1, \ldots, F_M$ and $G_1, \ldots, G_M$.

Similar to the previous section, we  first show that the ${H}_2$ suboptimal output containment control problem of the multi-agent system \eqref{hleaders_dynamic} and \eqref{hfollowers_dynamic} can be recast into a number of ${H}_2$ suboptimal control problems of a set of independent systems.

To proceed, we introduce the following $M$ auxiliary linear systems:
\begin{equation*}\label{auxiliary_mas}
	\begin{aligned} 
		\dot{\widetilde{x}}_i\left(t\right) & = A_i \widetilde{x}_i\left(t\right)  +  B_i \widetilde{u}_i\left(t\right)  +  E_i \widetilde{d}_i\left(t\right),\\
		\widetilde{\zeta}_i\left(t\right) & = C_{1i} \widetilde{x}_i\left(t\right)  +  D_{1i} \widetilde{d}_i\left(t\right),\\
		\widetilde{\epsilon}_i\left(t\right) & =  C_{2i} \widetilde{x}_i\left(t\right)  +  D_{2i} \widetilde{u}_i\left(t\right),
	\end{aligned}i = 1,\dots, M,
\end{equation*}
where $\widetilde{x}_i \in \mathbb{R}^{ni}$, $\widetilde{u}_i \in \mathbb{R}^{mi}$, $\widetilde{d}_i \in \mathbb{R}^{qi}$, $\widetilde{\zeta}_i  \in \mathbb{R}^{ri}$, and $\widetilde{\epsilon}_i \in \mathbb{R}^p$ are, respectively, the state, the coupling input, the external disturbance, the measured output, and the output to be controlled of the $i$th auxiliary system. By using the dynamic feedback controllers
\begin{equation}\label{auxiliary_dynamicprotocol}
	\begin{aligned}
		\dot{\widetilde{w}}_i & =A_i \widetilde{w}_i +  B_i  \widetilde{u}_i   +  G_i( C_{1i}\widetilde{w}_i - \widetilde{\zeta}_i) ,\\
		\widetilde{u}_i & = F_i \widetilde{w}_i,\quad i = 1,\dots, M,
	\end{aligned}
\end{equation}
the closed-loop auxiliary systems can be written as
	\begin{align}
		\begin{bmatrix}
			\dot{\widetilde{x}}_i \\
			\dot{\widetilde{w}}_i
		\end{bmatrix} 
		&=
		\begin{bmatrix}
			A_i &   B_iF_i \\
			 -G_iC_{1i} & A_i + B_iF_i + G _iC_{1i}
		\end{bmatrix} 
		\begin{bmatrix}
			\widetilde{x}_i\\
			\widetilde{w}_i
		\end{bmatrix}\\
  &\qquad +
		\begin{bmatrix}
			E_i \\
			-G_iD_{1i}
		\end{bmatrix} \widetilde{d}_i
		, \nonumber
		\\
		\widetilde{\epsilon}_i &= 
		\begin{bmatrix}
				C_{2i} &   D_{2i} F_i
		\end{bmatrix}
		\begin{bmatrix}
			\widetilde{x}_i\\
			\widetilde{w}_i
		\end{bmatrix},\quad i = 1,\dots, M. \label{decoupled_containment_dynamic}
	\end{align}
Denote $\widetilde{A}_i = \begin{bmatrix}
	A_i &   B_iF_i \\
	-G_iC_{1i} & A_i + B_iF_i + G _iC_{1i} 
\end{bmatrix}$, $\widetilde{E}_i = \begin{bmatrix}
	E_i \\
    -G_i D_{1i}
\end{bmatrix}$ and $ \widetilde{C}_i= \begin{bmatrix}
	 C_{2i} &   D_{2i} F_i
\end{bmatrix}$. Therefore, the impulse response matrix of each system~\eqref{decoupled_containment_dynamic} from the disturbance $\widetilde{d}_i$ to the output $\widetilde{\epsilon}_i$ is equal to
	$\widetilde{T}_{i, F_i, G_i}(t) = \widetilde{C}_i e^{\widetilde{A}_it}\widetilde{E}_i$ and the associated $H_2$ cost functional is defined as
		$J_i(F_i, G_i) := \int_{0}^{\infty} \text{tr}\left[ \widetilde{T}_{i, F_i, G_i}^{\top}(t) \widetilde{T}_{i, F_i, G_i}(t)\right] dt, i = 1,\dots, M.$

The following theorem shows that $H_2$ suboptimal output containment control is achieved if $H_2$ suboptimal
control problems of $M$ auxiliary systems are solved.
\begin{thm}\label{thm1}
	For $i = 1,\dots, M$. Let $F_i \in \mathbb{R}^{m_i \times n_i}$, $G_i \in \mathbb{R}^{n_i \times r_i}$ and $\gamma > 0$ be given. 
    The dynamic protocol~\eqref{heter_dynamicprotocol} with gain matrices $F_i$ and $G_i$ achieves output containment control for the multi-agent system \eqref{hleaders_dynamic} and  \eqref{hfollowers_dynamic} if and only if the dynamic controllers \eqref{auxiliary_dynamicprotocol} with $F_i$ and $G_i$ internally stabilize all the $M$ closed-loop systems \eqref{decoupled_containment_dynamic}. Furthermore, $J(F, G) < \gamma$ if $\sum_{i = 1}^{M} J_i(F_i, G_i) < \frac{\gamma}{\lambda_{M}^2},$ where $\lambda_M$ is the largest eigenvalue of the matrix $L_1$.
\end{thm}	
\begin{proof}
    We first show that the output containment control for the multi-agent system \eqref{hleaders_dynamic} and  \eqref{hfollowers_dynamic} is achieved if and only if the dynamic controllers internally stabilize all the $M$ closed-loop systems \eqref{decoupled_containment_dynamic}.
    
    Note that the output containment control is achieved if and only if the matrix ${A}_o$ of the controlled error system~\eqref{compact_heter_containment} is Hurwitz, i.e., the matrices $A + BF$, $A + GC_1$, and $I_M \otimes S - L_1 \otimes I_r$ are Hurwitz. Now, recall that the eigenvalues of $S$ lie on the imaginary axis and $L_1$ has only positive real eigenvalues, then the matrix $I_M \otimes S - L_1 \otimes I_r$ is Hurwitz. What remains to show is that $A + BF$ and $A + GC_1$ are Hurwitz, or alternatively, $A_i + B_iF_i$, $A_i + G_iC_{1i}$ for $i \in \mathcal{F}$ are Hurwitz. 
    
    Now, we look at the auxiliary systems \eqref{decoupled_containment_dynamic}. We show that the auxiliary systems are internally stable if and only if $A_i + B_iF_i$, $A_i + G_iC_{1i}$ for $i \in \mathcal{F}$ are Hurwitz. By applying the state transformation 
    	$\begin{bmatrix}
    		\widetilde{e}_i \\
    		\widetilde{x}_i
    	\end{bmatrix} 
    	=
    	\begin{bmatrix}
                I_{ni} & -I_{ni}\\
    		I_{ni}  & 0 
    	\end{bmatrix}
    	\begin{bmatrix}
    		\widetilde{x}_i \\
    		\widetilde{w}_i
    	\end{bmatrix}$ on the auxiliary systems \eqref{decoupled_containment_dynamic}, the matrices $\widetilde{A}_i = \begin{bmatrix}
    	A_i &   B_iF_i \\
    	G_iC_{1i} & A_i + B_iF_i + G _iC_{1i} 
    \end{bmatrix}$ can be transformed into $\widetilde{A}_{ei}: = \begin{bmatrix}
    		A_i + G _iC_{1i}&  0 \\
    		- B_iF_i & A_i + B_iF_i 
    \end{bmatrix}$ and likewise obtain $\widetilde{C}_{ei} := \begin{bmatrix}
    	 -D_{2i} F_i &  C_{2i} + D_{2i} F_i
    \end{bmatrix}, \widetilde{E}_{ei} := \begin{bmatrix}
    	E_i + G_i D_{1i} \\
        E_i
    \end{bmatrix}$. It is easy to see that all the
$M$ closed-loop systems \eqref{decoupled_containment_dynamic} are internally stable if and only if the matrices $A_i + B_iF_i$ and $A_i + G _iC_{1i}$ for $i = 1,\dots, M$  are Hurwitz. It then
follows from the above that output containment control for the multi-agent system \eqref{hleaders_dynamic} is achieved if and only if the
$M$ auxiliary closed-loop systems are internally stabilized by the controllers \eqref{auxiliary_dynamicprotocol}.

Next, let $F_i$ and $G_i$ be such that matrices $A_i + B_iF_i$, $A_i + G _iC_{1i}$ are Hurwitz for $i = 1,\dots, M$. We prove that if $\sum_{i = 1}^{M} J_i(F, G) < \frac{\gamma}{\lambda_{M}^2}$ holds then $J(F, G) < \gamma$. Recall that $J_i(F_i, G_i) = \int_{0}^{\infty} \text{tr}\left[ \widetilde{T}_{i, F_i, G_i}^{\top}(t) \widetilde{T}_{i, F_i, G_i}(t)\right] dt$ and $\widetilde{T}_{i, F_i, G_i}(t) = \widetilde{C}_i e^{\widetilde{A}_it}\widetilde{E}_i$, then it holds that
\begin{equation}\label{Jcost_heter1}
		\begin{aligned}
			&\lambda_M^2\sum_{i = 1}^{M} J_i(F_i, G_i) \\
            = & \int_{0}^{\infty} \lambda_M^2 \sum_{i = 1}^{M}  \text{tr}( (\widetilde{C}_ie^{\widetilde{A}_it }\widetilde{E}_i)^{\top} (\widetilde{C}_ie^{\widetilde{A}_it }\widetilde{E}_i)) dt\\
             = &\int_{0}^{\infty}  \lambda_M^2\text{tr}\left[ (\widetilde{C}e
                ^{\widetilde{A} t}\widetilde{E})^{\top}(\widetilde{C}e^{\widetilde{A}t}\widetilde{E})\right] dt \\
            \geq & \int_{0}^{\infty} \text{tr}\left[ (\widetilde{C}_e e^{\widetilde{A}_e t}\widetilde{E}_e)^{\top}(L_1 \otimes I_p)^{\top}(L_1 \otimes I_p)(\widetilde{C}_e e^{\widetilde{A}_et}\widetilde{E}_e)\right] dt, 
	\end{aligned}
 \end{equation}
where $\widetilde{C} e^{\widetilde{A} t}\widetilde{E} = {\rm blockdiag}(\widetilde{C}_1 e^{\widetilde{A}_1 t}\widetilde{E}_1,\dots, \widetilde{C}_M e^{\widetilde{A}_M t}\widetilde{E}_M)= \widetilde{C}_e e^{\widetilde{A}_e t}\widetilde{E}_e = {\rm blockdiag}(\widetilde{C}_{e1} e^{\widetilde{A}_{e1} t}\widetilde{E}_{e1},\dots, \widetilde{C}_{eM} e^{\widetilde{A}_{eM} t}\widetilde{E}_{eM})$ and in the last step of \eqref{Jcost_heter1} we have used the fact that $\lambda_M I_{pM} \geq L_1 \otimes I_p$.
In fact, it is not difficult to see that
	\begin{equation}\label{Jcost_heter2}
		\begin{aligned}
			&\int_{0}^{\infty} \text{tr}\left[ (\widetilde{C}_e e^{\widetilde{A}_e t}\widetilde{E}_e)^{\top}(L_1 \otimes I_p)^{\top}(L_1 \otimes I_p)(\widetilde{C}_e e^{\widetilde{A}_et}\widetilde{E}_e)\right] dt \\  &= \int_{0}^{\infty} \text{tr}\left[ ({C}_o e^{{A}_o t}{E}_o)^{\top} ({C}_o e^{{A}_ot}{E}_o)\right] dt = J(F ,G). 
		\end{aligned}
	\end{equation}
Subsequently, it follows from \eqref{Jcost_heter1} and \eqref{Jcost_heter2} that
	\begin{equation*}
			J(F ,G)\leq \lambda_M^2\sum_{i = 1}^{M} J_i(F_i, G_i) <  \gamma.
	\end{equation*}
 This completes the proof.
\end{proof}

By applying  Theorem \ref{thm1}, the ${H}_2$ suboptimal output containment control problem for multi-agent system \eqref{hleaders_dynamic} and  \eqref{hfollowers_dynamic} can be recast into a number of ${H}_2$ suboptimal control problems of $M$ auxiliary closed-loop systems \eqref{decoupled_containment_dynamic} by dynamic controllers \eqref{auxiliary_dynamicprotocol}.

In the following theorem, we provide a design method to compute appropriate gain matrices $F_1, \ldots, F_M$ and $G_1, \ldots, G_M$ of the dynamic protocol~\eqref{heter_dynamicprotocol} that achieves $H_2$ suboptimal ouput containment control.
\begin{thm}\label{thm2}
Let $\gamma > 0$ be given. For $i = 1,\dots, M$, assume that $D_{1i}E_i^{\top}  =0$, $D_{2i}^{\top} C_{2i} =0$, $D_{1i} D_{1i}^{\top} =I_{ri}$ and ${D}_{2i}^\top {D}_{2i} =I_{mi}$. The dynamic protocol \eqref{heter_dynamicprotocol} with gain matrices $F_i$ and $G_i$ for $i \in \mathcal{F}$ achieves the output containment control for multi-agent system \eqref{hleaders_dynamic} and \eqref{hfollowers_dynamic}, and $J(F, G) < \gamma$ if there exists $P_i > 0$ and $Q_i > 0$ satisfying
	\begin{align}\label{thm2_p}
		& A_i^{\top} P_i  + P_i A_i^{\top} - P_iB_iB_i^\top P_i + C_{2i}C_{2i}^\top< 0, \\
		& A_i Q_i + Q_i A_i^{\top} - Q_i C_{1i}^{\top} C_{1i} Q_i + E_i E_i ^{\top}< 0, \label{thm2_q}\\
	    &{\rm tr} \left(C_{1i} Q_i P_i Q_iC_{1i}^\top  \right) + {\rm tr}\left(  C_{2i} Q_i  C_{2i}^\top  \right)  < \frac{\gamma}{M\lambda_{M}^2}. \label{thm2_gamma}
	\end{align}
\end{thm}
\begin{proof}
Recall that $D_{1i}E_i^{\top}  =0$, $D_{2i}^{\top} C_{2i} =0$, $D_{1i} D_{1i}^{\top} =I_{ri}$ and ${D}_{2i}^\top {D}_{2i} =I_{mi}$. Then \eqref{thm2_p} is equivalent to 
\begin{equation}\label{riccati_F}
\begin{aligned}
& (A_i+B_iF_i)^{\top}P_i  + P_i(A_i+ B_i F_i) \\
&
\qquad +(C_{2i} + D_{2i} F_i)^\top  (C_{2i} + D_{2i} F_i) < 0.
\end{aligned}
\end{equation}
where $F_i := -B_i^{\top}P_i$.

By \eqref{thm2_gamma}, for $\epsilon_i >0$  sufficiently small, we have $\sum_{i=1}^{M} \gamma_i <\frac{\gamma}{\lambda_{M}^2}$, where $\gamma_i:= [{\rm tr}\left(C_{1i} Q_i P_i Q_iC_{1i}^\top \right) + {\rm tr}\left(C_{2i} Q_i  C_{2i}^\top\right)] + \epsilon_i$.
By taking $\bar{A} = A_i$, $\bar{B} = B_i$, $\bar{C}_1 =  C_{1i}$, $\bar{D}_1 = D_{1i}$, $\bar{C}_2 =   C_{2i}$, $\bar{D}_2 = D_{2i}$, and $\bar{E} =  E_i$, since there exist $P_i > 0$, $Q_i  > 0$ satisfying \eqref{thm2_q}, \eqref{riccati_F}, and \eqref{thm2_gamma}, it follows from Lemma \ref{lem_1} that all $M$ auxiliary systems are internally stabilized by the dynamic controllers \eqref{auxiliary_dynamicprotocol} with $F_i := -B_i^{\top}P_i$, $G_i := -Q_iC_{1i} ^\top$ for all $i = 1,\dots, M$, and $J_i(F_i, G_i) < \gamma_i$. Therefore, $\sum_{i = 1}^{M} J_i(F_i, G_i) < \sum_{i = 1}^{M} \gamma_i < \frac{\gamma}{\lambda_{M}^2}$.

Consequently, since all  $M$ closed-loop systems \eqref{decoupled_containment_dynamic} are internally stable and $\sum_{i = 1}^{M} J_i(F_i, G_i) < \frac{\gamma}{\lambda_{M}^2}$, it can be concluded from Theorem \ref{thm1} that the dynamic feedback protocol \eqref{heter_dynamicprotocol} achieves output containment control for the multi-agent system \eqref{hleaders_dynamic}, \eqref{hfollowers_dynamic}, and $J(F, G)<\gamma$.
\end{proof}

Note that in Theorem \ref{thm2}, the assumptions $D_{1i} D_{1i}^{\top} =I_{ri}$ and ${D}_{2i}^\top {D}_{2i} =I_{mi}$ are made to simplify notation and can be smoothly relaxed to the regularity condition $D_{1i} D_{1i}^{\top} > 0$ and ${D}_{2i}^\top {D}_{2i} > 0$ as in Lemma \ref{lem_1}.

\begin{rem}\label{remark1}
	Theorem \ref{thm2} states that the gains $F_i = -B_i^\top P_i$ and $G_i= -Q_iC_{1i} ^\top$ of dynamic protocol \eqref{heter_dynamicprotocol} is suboptimal by choosing feasible $P_i$ and $Q_i$ for $i \in \mathcal{F}$. Then the question arises: how do we find the smallest upper bound $\gamma$ such that ${\rm tr} \left(C_{1i} Q_i P_i Q_iC_{1i}^\top  \right) + {\rm tr}\left(  C_{2i} Q_i  C_{2i}^\top  \right)  < \frac{\gamma}{M\lambda_{M}^2}$ for $i = 1,\dots, M$?
    It is easy to note that in general the smaller $P_i$ and $Q_i$ lead to smaller ${\rm tr}\left(C_{1i} Q_i P_i Q_iC_{1i}^\top \right)$ and $ {\rm tr}\left(C_{2i} Q_i  C_{2i}^\top\right)$, consequently, the smaller $\gamma$. To find the feasible $\gamma$ as small as possible, we could find $P_i$ and $Q_i$ as small as possible.
Indeed, we can compute $P_i$ satisfying \eqref{thm2_p} by finding the solution $P_i(\delta_i) > 0$ with $\delta_i > 0 $ which satisfies
	\begin{equation*}
         A_i^{\top} P_i  + P_i A_i^{\top} - P_iB_iB_i^\top P_i + C_{2i}C_{2i}^\top + \delta_i I_{ni} = 0. 
	\end{equation*}
and similarly, $Q_i$ can be computed by find the solution $Q_i(\eta_i) > 0 $ with $\eta_i > 0$ which satisfies 
	\begin{equation*}
		A_i Q_i + Q_i A_i^{\top} - Q_i C_{1i}^{\top} C_{1i} Q_i + E_i E_i ^{\top}+\eta_i I_{ni} = 0.
	\end{equation*}
	Consequently, we can find the smallest solutions $P_i$ and $Q_i$ by choosing $\delta_i > 0$ and $\eta_i > 0$ very close to $0$. Therefore, we find the ‘best’ small $\gamma$ in the sense as explained above.
\end{rem}
\begin{rem}\label{remark2}
As a final remark, it is worth noting that the assumption regarding all eigenvalues of $S$ lying on the imaginary axis can be relaxed. More concretely, this assumption can be eliminated by considering the following protocol: 
\begin{equation*}
	\begin{aligned}
            \dot{w}_i & = A w_i + B u_i + G_i\left(C_{1i}w_i - y_i \right),\\
            \dot{v}_i & = S v_i + K\left(\sum_{j=1}^M a_{ij} (v_j - v_i) + \sum_{j={M+1}}^N a_{ij} (x_j - v_i)\right),\\
		u_i & = F_i(w_i - \Pi_i v_i) + \Gamma_i v_i, \qquad i \in \mathcal{F},
	\end{aligned}
\end{equation*} 
where $K$ is an additional gain matrix to be designed, see, e.g., \cite{jiao2016distributed}. To compute a suitable gain $K$ of the auxiliary reference system is similar to finding a gain matrix of static state feedback control protocol for solving homogeneous containment control problem, see, e.g., \cite{li2013distributed,yuan2023}.
\end{rem}

\section{Simulation Examples}\label{sec_examples}
In  this  section,  two  simulation  examples  are  provided  to  validate  the  effectiveness  of  our proposed protocols.
\subsection{Example for state containment control of homogeneous linear multi-agent systems}
In this subsection, we   use a simulation example to illustrate the performance of our designed dynamic protocol \eqref{containment_dynamicprotocol} to achieve $H_2$ suboptimal state containment control. 
Consider a homogeneous multi-agent system consisting of three leaders of the form \eqref{containment_leadersdynamic} and six followers of the form   \eqref{containment_followersdynamic}, where 
$A =
\begin{bmatrix}
	0 & 0 & -1 \\
	0 & 0 & 2 \\
	1 & 0 & -1.5 \\
\end{bmatrix}$,
$B = 
\begin{bmatrix}
	1 \\
	1.2\\
	1.5
\end{bmatrix},$
$E = 
\begin{bmatrix}
	0.1 & 0 & 0\\
	0.1 & 0 & 0\\
	0 & 0 & 0.1
\end{bmatrix},$
$C_1 = 
\begin{bmatrix}
	1 & 1 & 0
\end{bmatrix},$
$D_1 =
\begin{bmatrix}
	0 & 1& 0
\end{bmatrix}$,
$C_2 = 
\begin{bmatrix}
	0.2 & 0.2 &0.2 \\
	0.2 & 0.2 & 0.2\\
	0 & 0 & 0
\end{bmatrix},$
$D_2 =
\begin{bmatrix}
	0 & 0 & 1
\end{bmatrix}^\top$. 
The pair $(A, B)$ is stabilizable and the pair$(C_1, A)$ is detectable. We also have ${D_1} {E}^{\top}  = 
\begin{bmatrix}
	0 & 0 & 0
\end{bmatrix}$, ${D_2}^{\top}{C_2}  = 
\begin{bmatrix}
	0 & 0 & 0
\end{bmatrix}$, ${D_2}^{\top}{D_2}=1,$
and  ${D_1} {D_1}^\top =1$.
\begin{figure}[t]
	\centering
	\includegraphics[height=4cm]{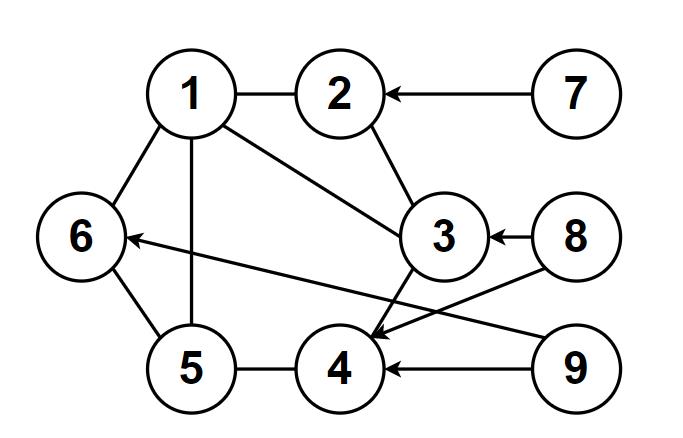}
	\caption{Communication topology between the leaders and the followers} \label{graph4.31}
\end{figure}

For illustration, let the communication graph $\mathcal{G}$ be given by Figure \ref{graph4.31}, where nodes 7, 8, and 9 denote the leaders and the other nodes denote the followers. Correspondingly, due to the particular form of the Laplacian matrix associated with $\mathcal{G}$,  the matrix $L_1$ of the Laplacian matrix is 
\begin{equation}\label{L1}
	L_1 =
	\begin{bmatrix}
		4  &  -1  &  -1  &   0  &   -1  &  -1  \\
		-1  &   3  &  -1   &  0  &   0   &  0  \\
		-1  &  -1   &  4  &  -1  &   0   &  0  \\
		0   &  0  &  -1   &  4  &  -1  &   0  \\
		-1  &   0  &   0  &  -1  &  3   & -1   \\
		-1  &   0   &  0   &  0   & -1   &  3  \\
	\end{bmatrix}.
\end{equation}
The smallest and largest eigenvalue of $L_1$ are computed to be $\lambda_1 = 0.6856$ and $\lambda_6 = 5.8245$. Now the proposed method in Theorem~\ref{thm8} is used to design a dynamic feedback protocol for solving the state containment control problem while satisfying $J(F, G) < \gamma$. Let the desired upper bound for the cost be $\gamma = 289$. We first compute a solution $P > 0$ from Theorem \ref{thm8} in \eqref{thm8_c2_p} by solving
\begin{equation*}\label{4.3p_soluiton}
	A^{\top} P + P A + ({c_p}^2 \lambda_{1}^4 - 2 {c_p} \lambda_{1}) PB B^{\top} P +\lambda_{6}^4 {C_2}^{\top} {C_2} + \delta I_3 =0
\end{equation*}
with $\delta = 0.001$ and  ${c_p} = \frac{2}{(\lambda_1+\lambda_6)(\lambda_1^2+\lambda_{6}^2)}=0.0089$ and compute a solution $Q > 0 $ in \eqref{thm8_c1_q} by solving
\begin{equation*}\label{4.3q_solution}
	A Q + Q A^{\top} - Q C_1^{\top} C_1 Q + \frac{1}{\lambda_1^2}E E ^{\top} + \eta I_3 = 0 
\end{equation*} 
with $\eta = 0.001$, which are the `best' choice to find a small upper bound $\gamma$ in the sense as explained in Remark \ref{remark4}. Then, with the command \texttt{icare} in Matlab
we compute the gain matrices $F = -{c_p}B^\top P = [-0.9439\ -0.7750\ -0.6738]$ and $G =-Q{C_1}^\top = [-0.0502\ -0.3429\ -0.0337]^\top$. Moreover, we compute $6({\rm tr}({C_1}QPQ{C_1}^\top) + \lambda_{6}^4{\rm tr}({C_2}Q{C_2}^\top)) = 288.2621
$, which is indeed smaller than the upper bound $\gamma = 289$.

Then by using the command \texttt{norm(sys,2)} in Matlab, the actual ${H}_2$ norm of the controlled error system \eqref{compact_containment_dynamic} is computed to be
\begin{equation*}
	||T_{F, G}||_{{H}_2} = 3.2966,
\end{equation*}
which is indeed smaller than $ \sqrt{\gamma} = \sqrt{289} = 17$.

Now we compare the performance of our protocol with that of the protocol proposed in \cite{li2017cooperative}, where the associated actual ${H}_2$ norm of the controlled system \eqref{compact_containment_dynamic} is computed as
\begin{equation*}
	||T_{K, L}||_{{H}_2} = 32.9022.
\end{equation*}
This result indicates that the performance of the dynamic protocol in \cite{li2017cooperative} is not comparable to that of our proposed dynamic protocol since its associated actual ${H}_2$ norm is much bigger than that of our protocol, i.e., $||T_{K, L}||_{{H}_2} = 32.9022 > ||T_{F, G}||_{{H}_2} = 3.2966$.

To statistically analyze the performance of our protocol in comparison to the protocol in \cite{li2017cooperative}, we introduced the disturbance matrix $E =
\begin{bmatrix}
e_1 & 0 & e_4\\
e_2 & 0 & e_5\\
e_3 & 0 & e_6
\end{bmatrix}$ with random values $||e_i||\leq 0.5, i= 1, \dots, 6$, which also satisfies the condition ${D_1} {E}^{\top} =
\begin{bmatrix}
0 & 0 & 0
\end{bmatrix}$. By involving 100 random sets of $E$, we proceeded to compare the actual ${H}_2$ norm of the controlled system \eqref{compact_containment_dynamic} using both protocols, and the results are depicted in Figure~\ref{figure_statistical}. The figure illustrates that our protocol outperforms the protocol in \cite{li2017cooperative}, in the sense that  the actual ${H}_2$ norm when using our protocol is significantly smaller than that using the protocol in \cite{li2017cooperative} for all random sets.
\begin{figure}[t]
	\centering
	\includegraphics[height=6cm]{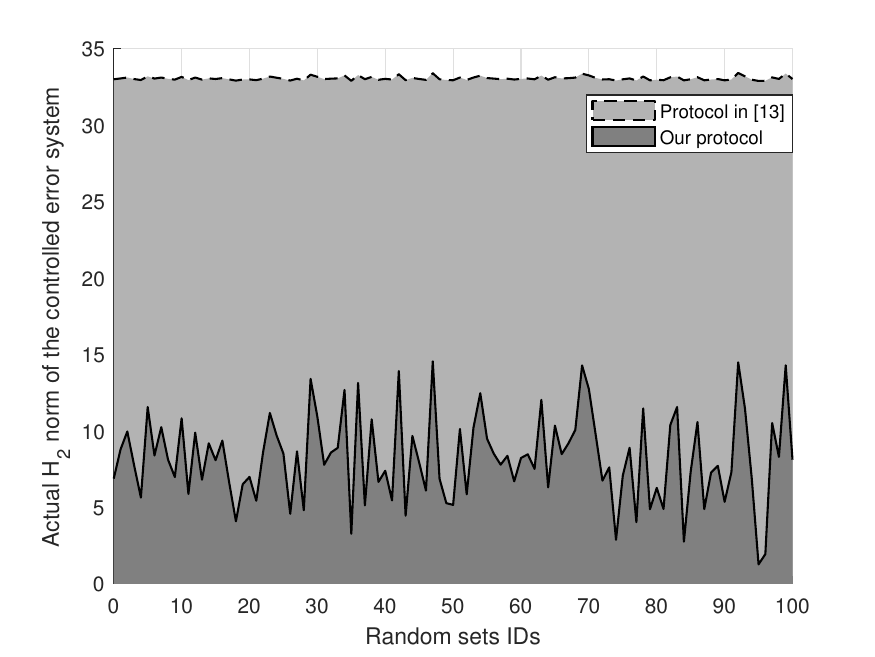}
	\caption{Actual ${H}_2$ norm of the controlled system by using two protocols with $100$ random sets of disturbance matrix $E$} \label{figure_statistical}
\end{figure}
\begin{figure*}[t]
	\begin{minipage}[t]{0.49\textwidth}
		\centering
		\includegraphics[height=7.5cm,width=9.5cm]{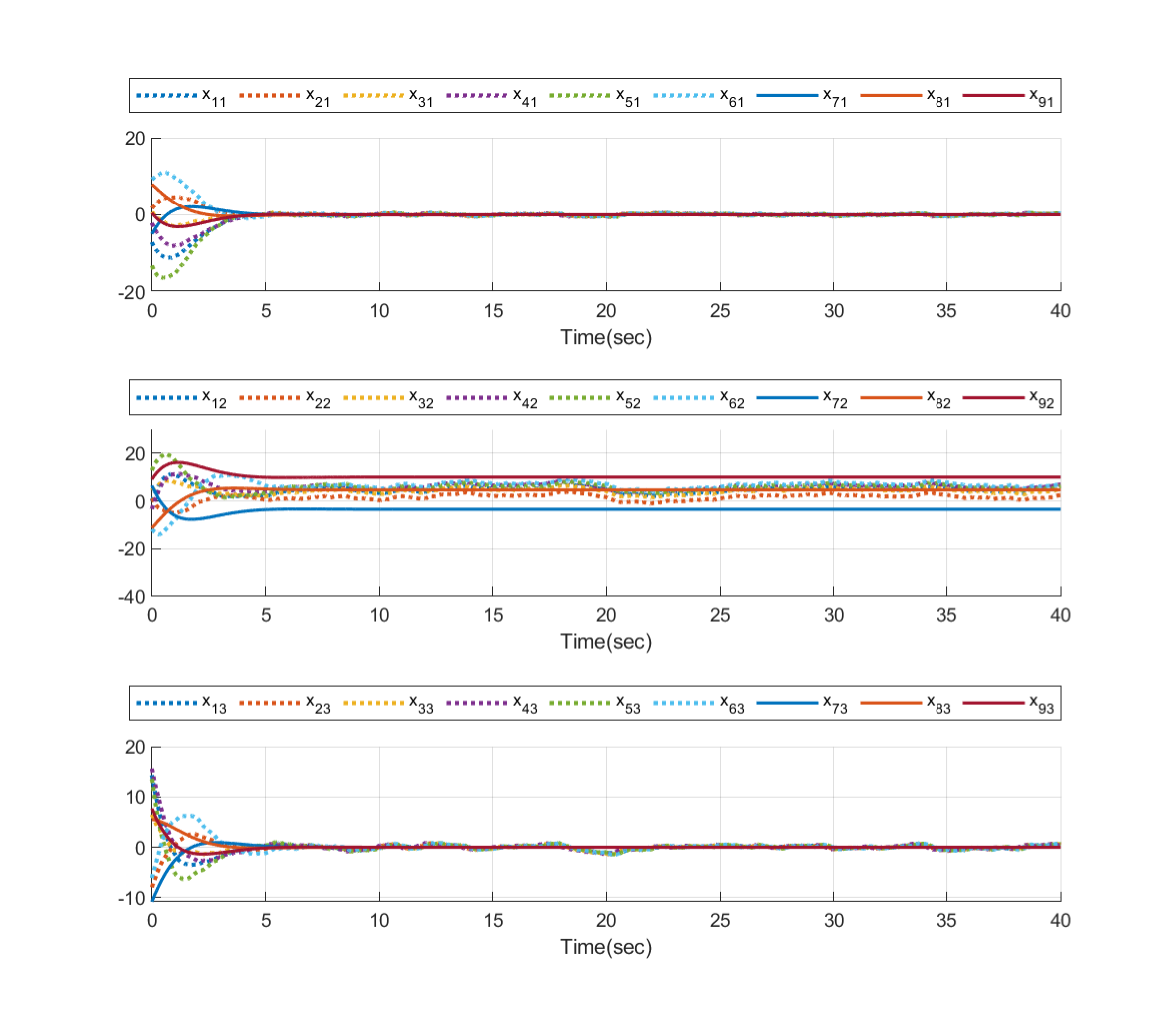}
	\end{minipage}
	\begin{minipage}[t]{0.49\textwidth}
		\centering
		\includegraphics[height=7.5cm,width=9.5cm]{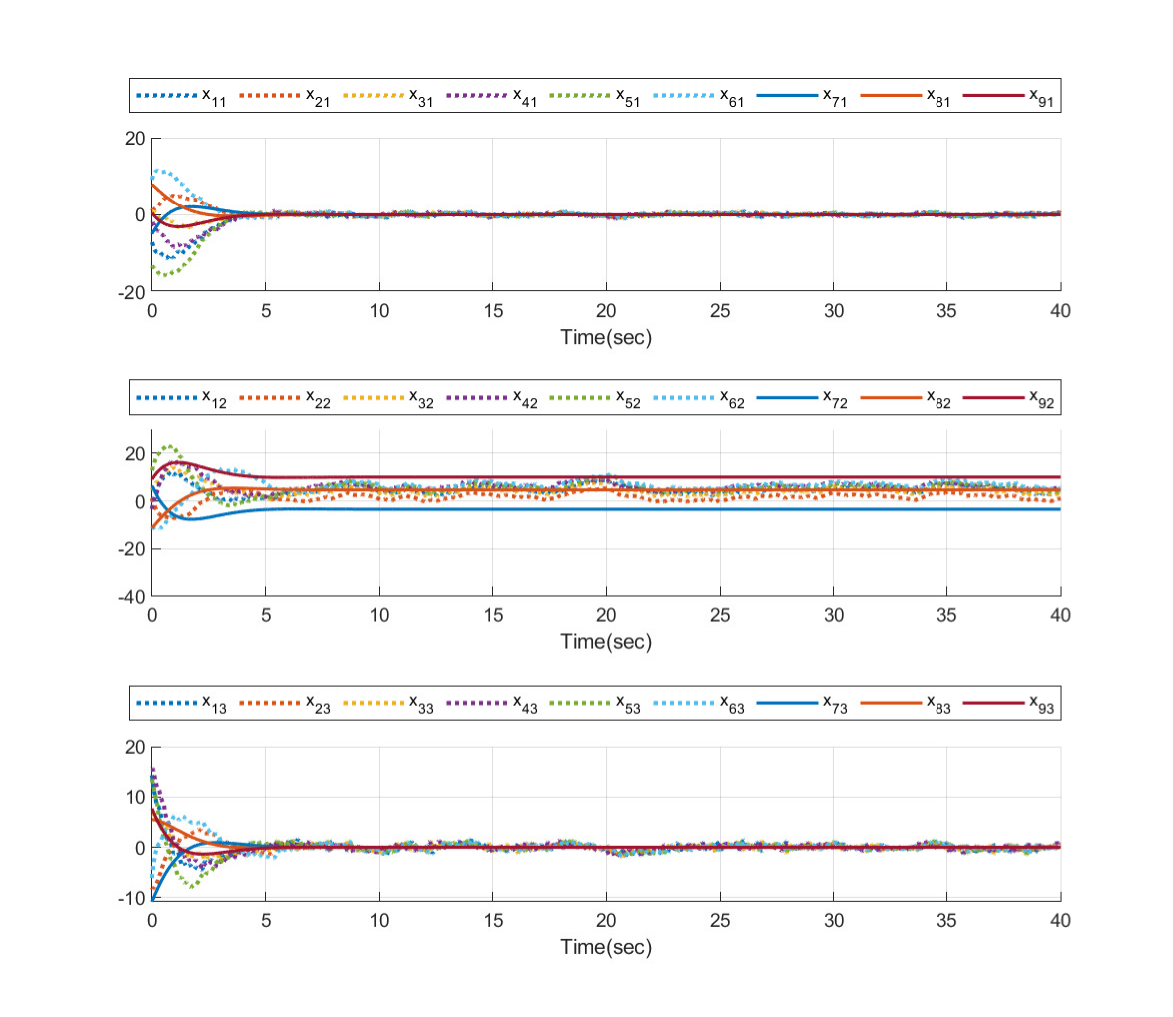}
	\end{minipage}
	\caption{Trajectories of the state $\bm{x_f}$ and $\bm{x_l}$ (with disturbances)  using our proposed protocol (left plot) and using the protocol in \cite{li2017cooperative} (right plot)}
	\label{figure4.31}
\end{figure*}
\begin{figure*}[t]
	\centering
	\begin{minipage}[t]{0.49\textwidth}
		\includegraphics[height=7.5cm,width=9.5cm]{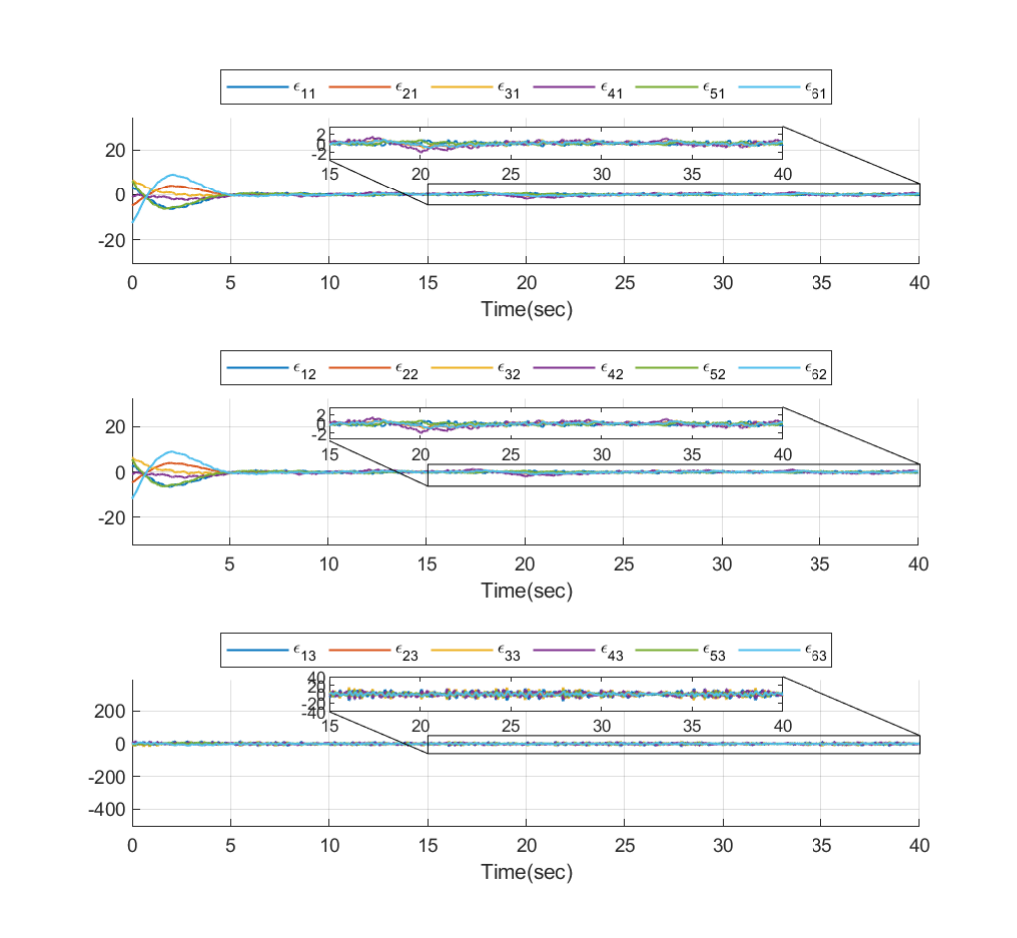}
	\end{minipage}
	\begin{minipage}[t]{0.49\textwidth}
		\includegraphics[height=7.5cm,width=9.5cm]{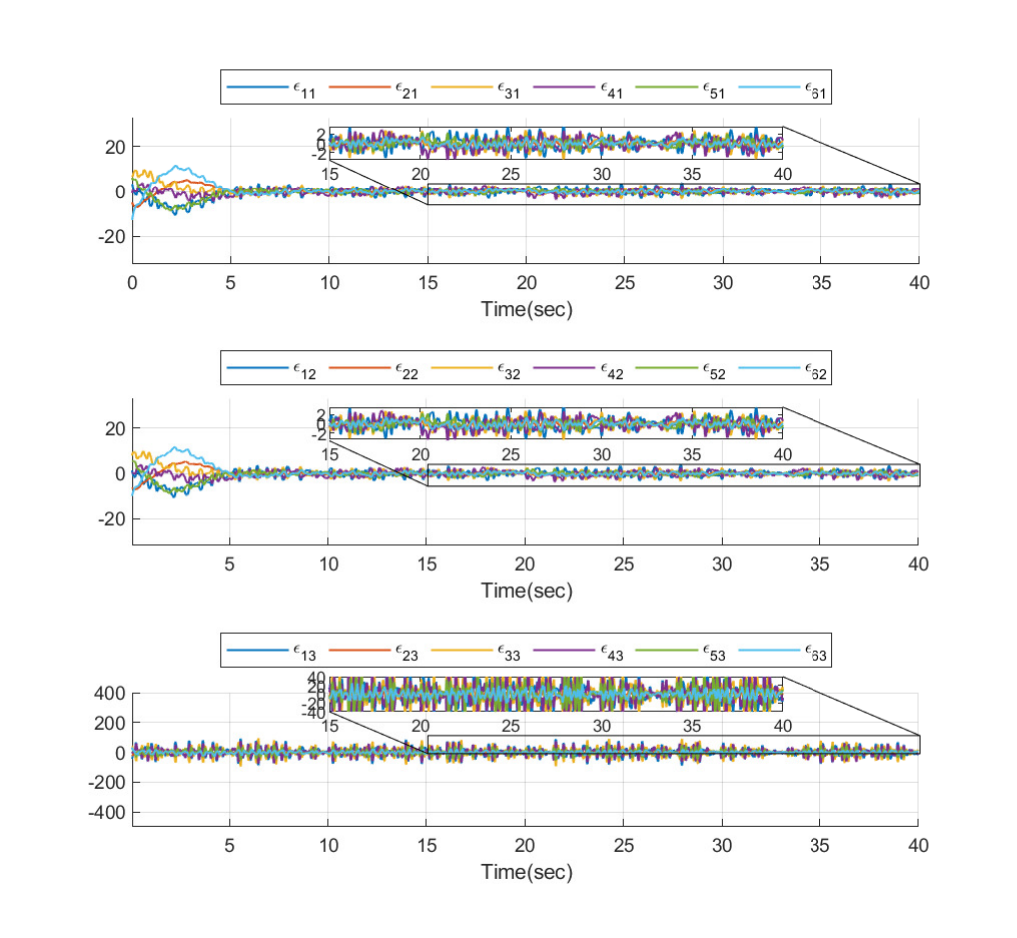}
	\end{minipage}
	\caption{ Trajectories of the performance output variables $\bm{\epsilon}$ (with disturbances) using our proposed protocol (left plot) and  using the protocol in \cite{li2017cooperative} (right plot)}
	\label{figure4.32}
\end{figure*}

As an illustrative example, we have randomly selected the initial states of the followers to be $x_{10} = [-7.09\ -0.11\ -14.33]^\top$, $x_{20} = [1.70\ 1.20\ -7.97]^\top$, $x_{30} = [0.74\ 4.5\ 6.47]^\top$,   $x_{40} = [-2.09\ -3.39\ 15.62]^\top$, $x_{50} = [-13.14\ 12.81\ 13.58]^\top$, $x_{60} = [9.18\ -11.76\ -6.11]^\top$, the initial states of the leaders to be $x_{70} = [-4.97\ 6.49\ -10.83]^\top$, $x_{80} = [7.98\ -11.29\ 5.5]^\top$, $x_{90} = [0.47\ 9.06\ 7.68]^\top$. Additionally, we set the initial values of the protocol state $\bm{w_f}$ to be zero and randomly choose the disturbance matrix to be $E = \begin{bmatrix}
0.25 & 0 & -0.21\\
0.19 & 0 & 0.07\\
-0.01 & 0 & 0.04 \end{bmatrix}$. 
To complete this scenario, a white noise signal $\bm {d}$ was applied, with amplitudes ranging between -15 and 15. The trajectories of the states $\bm{x_f}$ and $\bm{x_l}$ by using our designed dynamic protocol and the protocol in \cite{li2017cooperative}  are plotted in Figure~\ref{figure4.31}, and the corresponding trajectories of the performance output variable $\bm{\epsilon}$ are plotted in Figure \ref{figure4.32}. It can be seen that our protocol performs better than the dynamic protocol in \cite{li2017cooperative}, in the sense that our dynamic protocol has a better tolerance for external disturbances.

\subsection{Example for output containment control of heterogeneous linear multi-agent systems}
In this subsection, we  give a simulation example to illustrate the performance of our designed protocol \eqref{heter_dynamicprotocol}. 
Consider a heterogeneous multi-agent system consisting of $N = 9$ agents with three leaders of the form \eqref{hleaders_dynamic} and six followers of the form   \eqref{hfollowers_dynamic}, where 
$S = 
\begin{bmatrix}
	0 & 1 \\
	0 & 0
\end{bmatrix}$,
$R = 
\begin{bmatrix}
	1 & 1 \\
	0 & 1
\end{bmatrix}$,
$A_i =
\begin{bmatrix}
     0 & 1 & 0 \\
	 0 & 0 & c_i \\
	 0 & -f_i & -a_i \\
\end{bmatrix}$,
$B_i = 
\begin{bmatrix}
	0 \\
	0\\
	b_i
\end{bmatrix},$
$E_i = 
\begin{bmatrix}
	0 & 0.2\\
	0 & 0\\
	0 & 0.2
\end{bmatrix},$
$C_{2i} = 
\begin{bmatrix}
	1 & 1 & 0 \\
	0 & 0 & 0
\end{bmatrix}$,
$D_{2i} = 
\begin{bmatrix}
	 0 \\
	 1 
\end{bmatrix}$, $C_{1i} = 
\begin{bmatrix}
	1 & 2 & 1 \\
\end{bmatrix},$
$D_{1i} =
\begin{bmatrix}
	1 & 0
\end{bmatrix}$. The parameters $a_i, b_i, c_i$ and $f_i$ are chosen as 
\begin{equation*}
\begin{aligned}
a_i &= 2, ~c_i = 1, \quad i = 1,\dots, 6,\\
b_1 &= b_4 = 1, ~b_2 = b_5 = 2, ~b_3 = b_6 = 3,\\
f_1 &= f_4 = 1, ~f_2 = f_5 = 2, ~f_3 = f_6 = 3.
\end{aligned}
\end{equation*}
\begin{figure}[t]
	\centering
	\includegraphics[height=6cm]{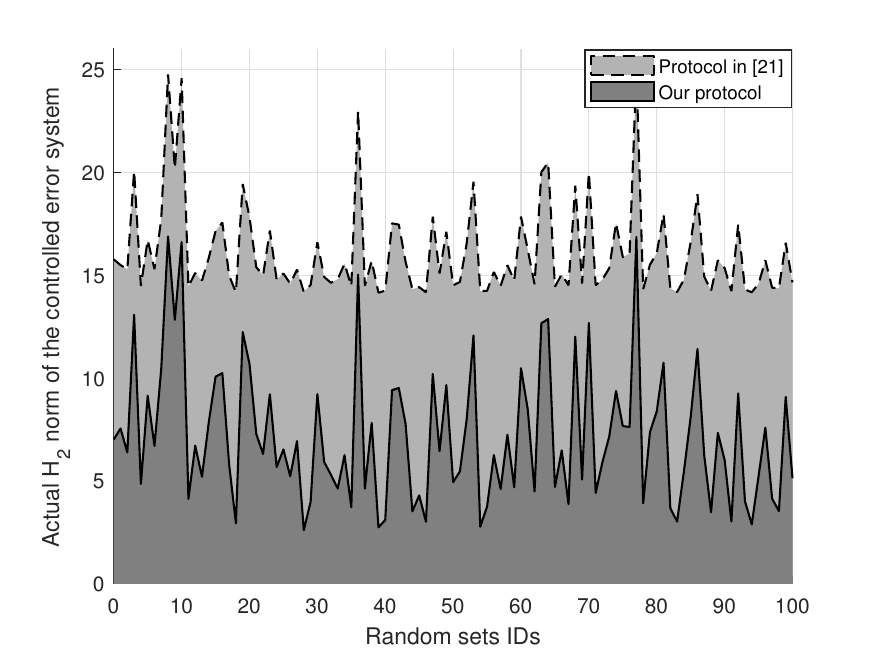}
	\caption{Actual ${H}_2$ norm of the controlled system by using two protocols with $100$ random sets of disturbance matrices $E_i$ for $i = 1,\dots, 6$.} \label{figure_statistical1}
\end{figure}
It is easy to check that the pair $\left(R, S\right)$ is observable, the eigenvalues of $S$ lie on the imaginary axis. 
 the pairs $({A_i}, {B_i})$ are stabilizable, and the pairs $({C_{1i}}, {A_i})$ are detectable.  We also have ${D_{1i}} {E_i}^{\top}  = 
\begin{bmatrix}
	0 & 0
\end{bmatrix}$, ${D_{2i}}^{\top}{C_{2i}}  = 
\begin{bmatrix}
	0 & 0 & 0
\end{bmatrix}$, ${D_{2i}}^{\top}{D_{2i}}=1,$ ${D_{1i}} {D_{1i}}^\top =1$, and the solutions of equations \eqref{regulation} are computed to be $\Pi_i = \begin{bmatrix}
	1 & 0\\
	0 & 1\\
	0 & 0
\end{bmatrix}$, $\Gamma_i = \begin{bmatrix} 1 & 0
\end{bmatrix}$ for $i = 1,\dots, 6$.

For illustration, the communication graph $\mathcal{G}$ between the agents is given by Figure \ref{graph4.31}, where nodes 7, 8, and 9 represent the leaders, and the other nodes represent the followers.
Correspondingly, the matrix $L_1$ of the Laplacian matrix is the same as \eqref{L1}.The largest eigenvalue of $L_1$ is computed to be $\lambda_6 = 5.8245$. 
Now we  use the method proposed in Theorem~\ref{thm2}  to design a   protocol \eqref{heter_dynamicprotocol} to solve the heterogeneous output containment problem while satisfying $J(F, G) < \gamma$. We let the desired upper bound for the cost be $\gamma = 115$. We use  Theorem \ref{thm2} to compute a solution $P_i > 0$  by solving 
\begin{equation*}\label{p_solution}
	A_i^{\top} P_i  + P_i A_i^{\top} - P_iB_iB_i^\top P_i + C_{2i}C_{2i}^\top + \delta  I_{n} = 0
\end{equation*}
with $\delta = 0.001$ and compute a solution $Q_i > 0 $ by solving
\begin{equation*}\label{q_solution}
	A_i Q_i + Q_i A_i^{\top} - Q_i C_{1i}^{\top} C_{1i} Q_i + E_i E_i ^{\top}+\eta I_{n} = 0.
\end{equation*}
with  $\eta = 0.001$, which are sufficiently small to minimize the $H_2$ performance upper bound $\gamma$ in the sense as explained in Remark \ref{remark1}. Then, in Matlab with the command \texttt{icare}, we compute  the feedback gain matrices $F_i$ and $G_i$ for $i = 1,\dots, 6$ to be 
\begin{equation*}
\begin{aligned}
F_1 &= F_4 = [-1.0005\ -1.7329\ -0.7326],\\
F_2 &= F_5 = [-1.0005\ -1.2345\ -0.4951],\\
F_3 &= F_6 = [-1.0005\ -1.0327\ -0.3982],
\end{aligned}
\end{equation*}
and 
\begin{equation*}
\begin{aligned}
G_1 &= G_4 = [0.3051\ ~0.0433\ ~0.0157]^\top,\\
G_2 &= G_5 = [0.2712\ ~0.0265\ ~0.0121]^\top,\\
G_3 &= G_6 = [0.2534\ ~0.0187\ ~0.0111]^\top.
\end{aligned}
\end{equation*}
Furthermore, we compute $S_i:= [{\rm tr}\left(C_{1i} Q_i P_i Q_iC_{1i}^\top \right) +{\rm tr}\left(C_{2i} Q_i  C_{2i}^\top\right)]$ for $i = 1,\dots, 6$ and obtain that 
\begin{figure*}[t]
	\centering
	\begin{minipage}[t]{0.49\textwidth}
		\centering
		\includegraphics[height=7.5cm,width=9.5cm]{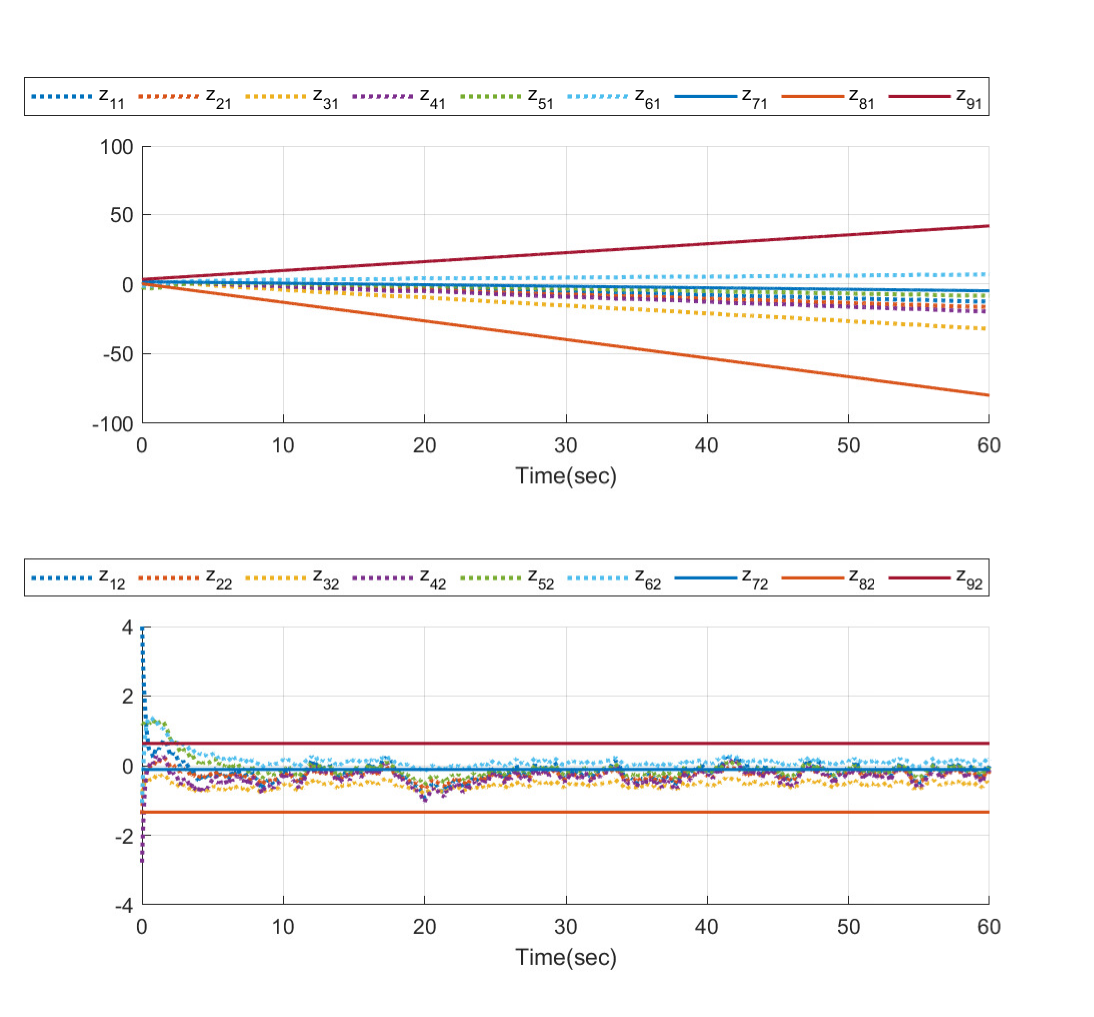}
	\end{minipage}
	\begin{minipage}[t]{0.49\textwidth}
		\centering
		\includegraphics[height=7.5cm,width=9.5cm]{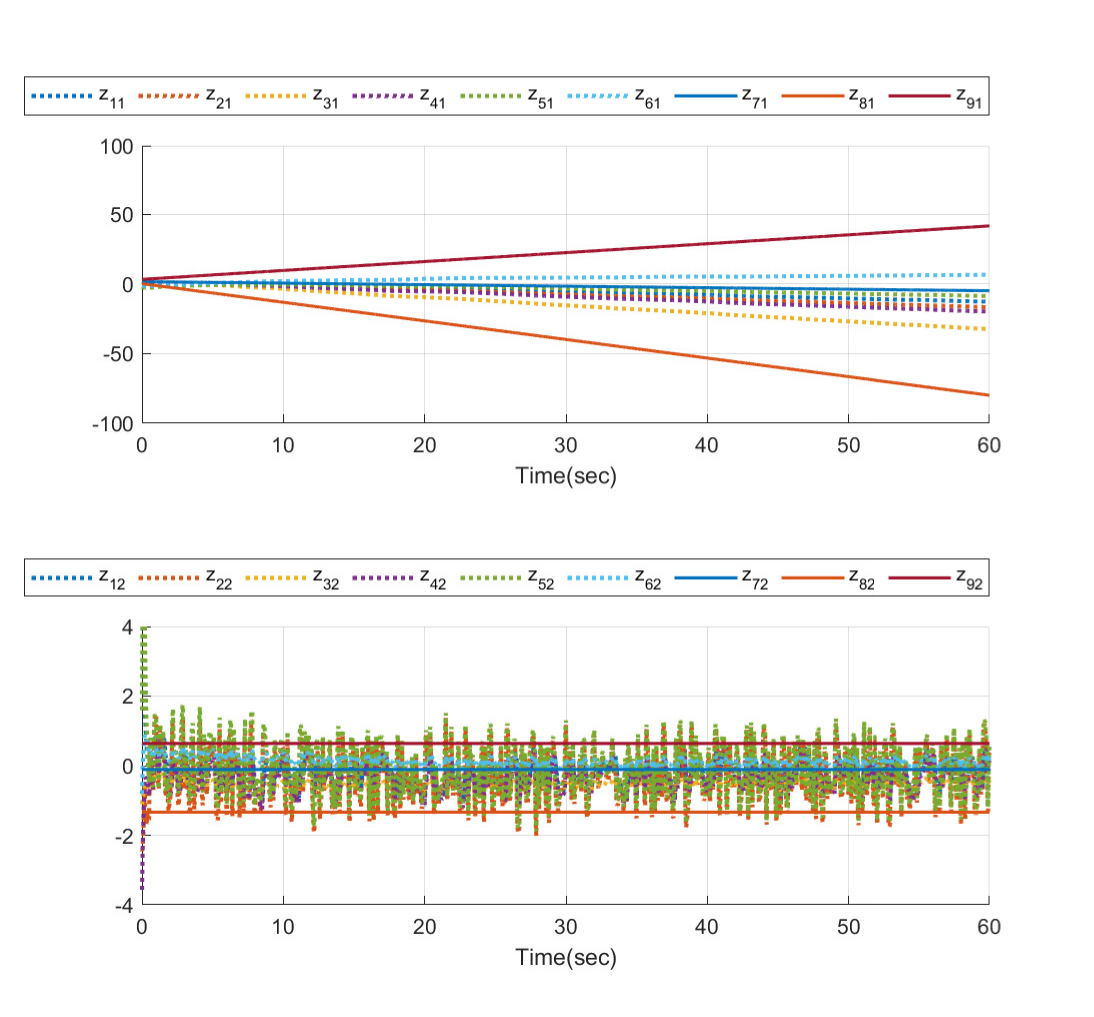}
	\end{minipage}
	\caption{Trajectories of the state $\bm{z_f}$ and $\bm{z_l}$ (with disturbances) using our proposed protocol (left plot) and using the protocol in \cite{qin2018output} (right plot)}
	\label{figure2}
\end{figure*}
\begin{figure*}[t]
	\centering
	\begin{minipage}[t]{0.49\textwidth}
		\centering
		\includegraphics[height=7.5cm,width=9.5cm]{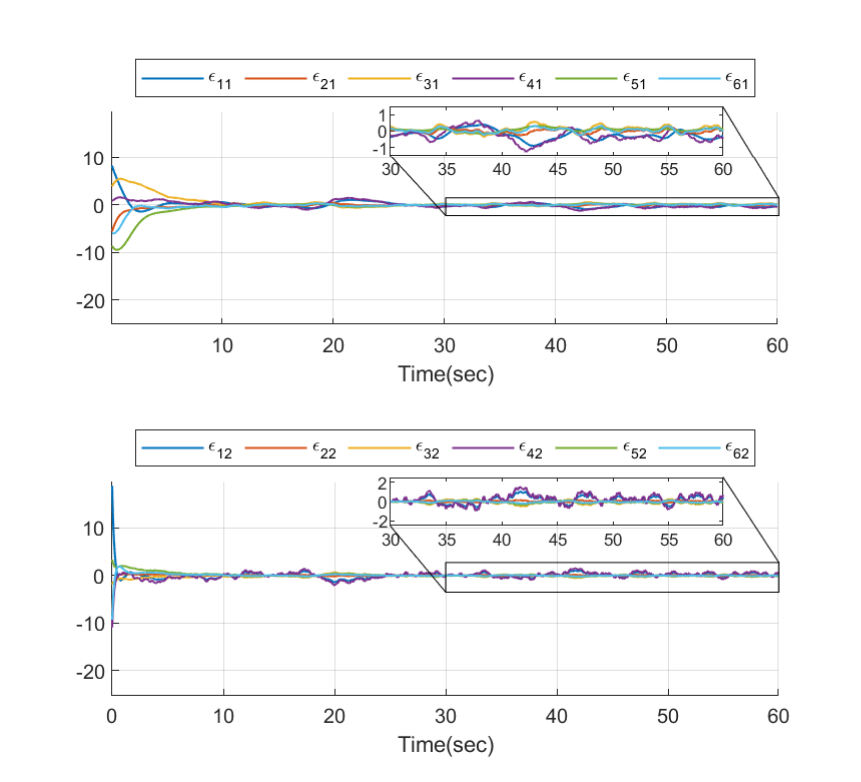}
	\end{minipage}
	\begin{minipage}[t]{0.49\textwidth}
		\centering
		\includegraphics[height=7.5cm,width=9.5cm]{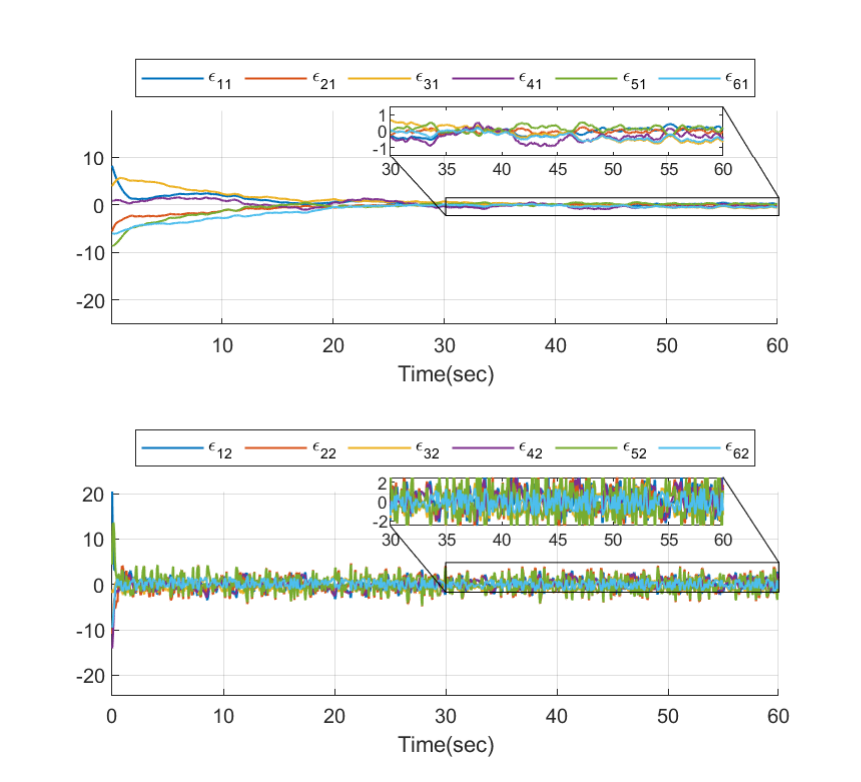}
	\end{minipage}
	\caption{ Trajectories of the performance output variables $\bm{\epsilon}$ (with disturbances) using our proposed protocol (left plot) and using the protocol in \cite{qin2018output} (right plot)}
	\label{figure3}
\end{figure*}%
\begin{equation*}
\begin{aligned}
S_1 = S_4 = 0.5630,~ S_2 = S_5 = 0.3917,~ S_3 = S_6 = 0.3350.
\end{aligned}
\end{equation*}
Note that, $S_i < \frac{\gamma}{6\lambda_{6}^2} = 0.5650$ for all $i = 1,\dots, 6$. Now, the actual ${H}_2$ norm of the controlled error system \eqref{compact_heter_containment} can be computed by using the command \texttt{norm(sys,2)} in Matlab as
\begin{equation*}
	||T_{F, G}||_{{H}_2} = 6.2937,
\end{equation*}
which is indeed smaller than $ \sqrt{\gamma} = \sqrt{115} = 10.7238$.

Next, we compare the performance of our   protocol with that of the proposed protocol in \cite{qin2018output}, where the associated actual ${H}_2$ norm of the controlled error system \eqref{compact_heter_containment} is computed to be
\begin{equation*}
	||\bar{T}_{F, G}||_{{H}_2} = 15.1388.
\end{equation*} It can be seen that the performance of the dynamic protocol in \cite{qin2018output} is not comparable to that of our dynamic protocol since its associated actual ${H}_2$ norm is much bigger than that of our protocol, i.e., $||\bar{T}_{F, G}||_{{H}_2} = 15.1388 > ||T_{F, G}||_{{H}_2} = 6.2937$.

To statistically analyze the performance of our protocol in comparison to the protocol in \cite{qin2018output}, we introduced the disturbance matrices $E_i =
\begin{bmatrix}
 0 & e_1\\
 0 & e_2\\
 0 & e_3
\end{bmatrix}$ with random values $||e_1||\leq 0.5,$  $||e_2||\leq 0.5$ and $||e_3||\leq 0.5$, Note that these disturbance matrices $E_i$ satisfies the conditions ${D_{1i}} {E_i}^{\top} =[0\ 0]$ for $i = 1,\dots, 6$. By introducing 100 random sets of $E_i$ for $i = 1,\dots, 6$, we compare the actual ${H}_2$ norm of the controlled system \eqref{compact_heter_containment} using both protocols. The comparison results are plotted in Figure~\ref{figure_statistical1}, which shows that the actual ${H}_2$ norm when using our protocol is significantly smaller than that using the protocol in  \cite{qin2018output} for all random sets. In other words, our protocol outperforms that proposed in \cite{qin2018output}.

As an illustrative example, we have randomly selected the initial states of the followers to be 
$x_{10} = [2.58\ -0.82\ -1.99]^\top$, 
$x_{20} = [ -0.99\ 0.49\ 2.28]^\top$, 
$x_{30} = [1.52\ -0.06\ 1.29]^\top$, 
$x_{40} = [2.12\ -1.23\ 1.59]^\top$, 
$x_{50} = [-0.51\ -1.62\ -1.72]^\top$,
$x_{60} = [-0.74\ -0.26\ -1.1]^\top$, the initial states of the leaders to be
$x_{70} =
[1.89\ -0.11]^\top$,
$x_{80} = [1.63\ -1.34]^\top$,
$x_{90} = [2.76\ 0.64]^\top$. We take the initial states $w_i$ to be zero, and the initial states $v_i$ to be $v_{10} = [-0.34\  1.67]^\top$, 
$v_{20} = [2.28\ -1.64]^\top$, $v_{30} = [0.14\ -0.46]^\top$, 
$v_{40} = [1.23\ -1.47]^\top$, 
$v_{50} = [-1.03\ 1.01]^\top$,
$v_{60} = [-0.54\ -0.26]^\top$. Additionally, we randomly choose the disturbance matrices to be $E_i =
\begin{bmatrix}
0 & 0.06\\
0 & 0.18\\
0 & 0.36
\end{bmatrix}$ for $i = 1, \dots, 6$, 
 and apply the same white noise $\bm d$ with an amplitude between $-2$ and $2$ to further compare the performance of our proposed protocol and of the protocol proposed in \cite{qin2018output}. In  Figure \ref{figure2}, we have plotted the trajectories of the output variable $\bm{z_f}$ and $\bm{z_l}$  using our designed protocol and the protocol in \cite{qin2018output}. In Figure \ref{figure3}, we have plotted the associated trajectories of the performance output variable $\bm{\epsilon}$.
The figures show that our proposed dynamic protocol performs better than the dynamic protocol in \cite{qin2018output} in the sense that our protocol has a better tolerance for external disturbances.

\section{Conclusion}\label{sec_conclusion}
In this paper, we have studied the ${H}_2$ suboptimal state containment control problem of homogeneous linear multi-agent systems and the   ${H}_2$ suboptimal output containment control problem of heterogeneous linear multi-agent systems. For both problems, we have proposed a distributed dynamic output feedback control protocol to achieve  $H_2$ suboptimal state/output containment control, i.e., the states/outputs of the followers converge to the convex hull spanned by the states/outputs of the leaders
and the associated ${H}_2$  cost is smaller than a given bound.

\balance

 \bibliographystyle{IEEEtran}   
	\bibliography{ref}  
\end{document}